\RequirePackage{fix-cm}
\documentclass[smallextended]{svjour3}       % onecolumn (second format)
\usepackage[utf8]{inputenc}
\usepackage[colorlinks,
            linkcolor=red,
            anchorcolor=red,
            citecolor=blue
            ]{hyperref}
\usepackage{graphicx,amsmath,fullpage,amssymb,amsthm,amsfonts,color}
\usepackage{enumerate}

\numberwithin{equation}{section}
\numberwithin{theorem}{section}
\numberwithin{definition}{section}
\numberwithin{corollary}{section}
\numberwithin{lemma}{section}
\numberwithin{remark}{section}

\theoremstyle{definition}

\usepackage[toc,page]{appendix} 
\usepackage[utf8]{inputenc}

\renewcommand{\epsilon}{\varepsilon}

\def\rd{{\rm d}}
\def\va{{\bf a}}
\def\vb{{\bf b}}

\def\vx{{\bf x}}
\def\mA{{\bf A}}
\def\mB{{\bf B}}

\def\mD{{\bf D}}

\def\vX{{\bf X}}

\def\MM{{\bf M}}
\def\vg{{\bf g}}
\def\vy{{\bf y}}

\def\vz{{\bf z}}
\def\vZ{{\bf Z}}

\def\vgamma{{\boldsymbol \gamma}}
\def\mmu{{\boldsymbol \mu}}

\def\mSigma{{\mathbf \Sigma}}
\def\mXi{{\mathbf \Xi}}

\begin{document}

\title{Stochastic limit-cycle oscillations of a nonlinear system under random perturbations
\thanks{H.Q. is partially supported by the Olga Jung Wan Endowed Professorship.
}
}

%\titlerunning{Short form of title}        % if too long for running head

\author{Yu-Chen Cheng \and Hong Qian}

%\authorrunning{Short form of author list} % if too long for running head

\institute{Yu-Chen Cheng \at Department of Applied Mathematics, University of Washington, Seattle, WA 98195-3925 \\
           \email{yuchench@uw.edu}         
               \and
           Hong Qian \at
              Department of Applied Mathematics, University of Washington, Seattle, WA 98195-3925   \\ 
              \email{hqian@uw.edu} 
              }

\date{Received: date / Accepted: date}
% The correct dates will be entered by the editor

\maketitle

\tableofcontents

\vskip 1cm

\begin{abstract}
Dynamical systems with $\epsilon$ small random perturbations appear in both continuous mechanical motions and discrete stochastic chemical kinetics. The present work provides a detailed analysis of the central limit theorem (CLT), with a time-inhomogeneous Gaussian process, near a deterministic limit cycle in $\mathbb{R}^n$. Based on the theory of random perturbations of dynamical systems and the WKB approximation respectively, results are developed in parallel from both standpoints of stochastic trajectories and transition probability density and their relations are elucidated.  We show rigorously the correspondence between the local Gaussian fluctuations and the curvature of the large deviation rate function near its infimum, connecting the CLT and the large deviation principle of diffusion processes. We study uniform asymptotic behavior of stochastic limit cycles through the interchange of limits of time $t\to\infty$ and $\epsilon\to 0$. Three further characterizations of stochastic limit cycle oscillators are obtained: (i) An approximation of the probability flux near the cycle; (ii) Two special features of the vector field for the cyclic motion; (iii) A local entropy balance equation along the cycle with clear physical meanings. Lastly and different from the standard treatment, the origin of the $\epsilon$ in the theory is justified by a novel scaling hypothesis via constructing a sequence of stochastic differential equations.

% \PACS{PACS code1 \and PACS code2 \and more}
\end{abstract}

\section{Introduction} \label{ch1}

{\em Newtonian mechanics} represents the world in terms of featureless point masses with their positions and momenta. In contrast, classical {\em chemical kinetics} represents the world in terms of the number densities of interacting populations of individual molecules, each with a large internal degrees of freedom, as chemical species.  What is possibly an appropriate representation for complex biological systems and processes?  To answer this question, it is necessary to to give a more precise meaning to the too widely used term  ``complex" \cite{hopfield1994physics}. Let us consider one class of complex systems, the living biological cells in terms of a biochemical kinetic description. In this paradigm, a {\em complex system} consists of many interacting sub-populations of individuals with stochastic state transitions; the system as a whole actively exchanges matters, energy, or information with its environment  \cite{van2019complexity}.  One sees a remarkable resemblance between this kinetic description of
cells and many other biological systems with complex ``individuals''.  In fact, the  biochemist's perspective captures a repeated hierarchical structure of the complex world: An ecological system is a community of various biological organism; a human body consists of over 30 trillion cells; and a cell involves a large number of interacting non-living biopolymers. This view echos the philosophy of P. W. Anderson's hierarchical structure of science \cite{anderson1972more}.

While the ``stochasticity'' in chemical kinetics mainly originates from internal states of individuals, uncertainties in mechanical motions in biology, such as protein motor proteins in axonal transport and hemodynamics of cardiovascular systems, are chiefly a consequence of {\em coarse graining}: A highly complex many-body systems can be represented by simple statistical laws.  One of the best examples of this is Kramers' rate theory for barrier crossing between two basins, which condenses a very complex dynamics into a simple exponentially distributed time with a single parameter. A problem becomes simple if we focus on the emergent behavior of an assembly of a large numbers of atoms at the macroscopic scale with a much longer time scale.  Indeed, experimentalists would  find that the macromolecular movement obeys simple laws under certain approximations, for example Fick's law. The bridge between complexity and simplicity is {\em uncertainty and its statistics}. This is the fundamental idea of the theory of Brownian motion
\cite{berg1993random}.  

\subsection{Stochastic models of complex systems} \label{ch1-1}
As can be seen from the above discussion, both representations have their own values for complex systems. Once we choose one of them to describe a system of interest, then the following question is what mathematical model we should adopt.
In stochastic chemical kinetics, there is a success of the well-established scaling hypothesis in the continuous-time non-negative integer valued Markov jump process $\mathbf{n}_V(t)$  \cite{kurtz1972relationship}. Consider a continuous stirred
chemical reaction vessel of volume $V$, in
which the number of molecules $\mathbf{n}_V(t)$
is a Markov jump process that 
can be described by a {\em master equation}
\begin{align} \label{master-equation}
    \frac{\partial P(\mathbf{n}_V, t)}{\partial t} = \sum_{\mathbf{r}} \left[W(\mathbf{n}_V - \mathbf{r}, \mathbf{r})P(\mathbf{n}_V -\mathbf{r}, t) -  W(\mathbf{n}_V , \mathbf{r})P(\mathbf{n}_V, t)  \right], 
\end{align}
where $W(\mathbf{n}_V , \mathbf{r})$ is the transition probability per unit time from $\mathbf{n}_V, \mathbf{n}_V + \mathbf{r}$, and both $\mathbf{n}_V$ and $\mathbf{r}$ are $q$-dimensional vectors. As the system's size $V \rightarrow \infty$,
$\mathbf{n}_V(t)$ follows the law of large number, $V^{-1}\mathbf{n}_V(t) \rightarrow \mathbf{c}(t)$, the concentration of $q$ species. 

With a proper scaling by the size $V$ and the assumption that in the limit $W$ and $P$ are smooth enough functions, we can take the Kramers-Moyal expansion of the master equation \eqref{master-equation} \cite{kramers1940brownian,moyal1949stochastic,gang1987stationary}
\begin{align} \label{KM-expan}
    \epsilon \frac{\partial p(\vx, t)}{\partial t} = \sum_{\mathbf{k}} \left(\frac{1}{\mathbf{k} !}\right) \left(\epsilon \frac{\partial}{\partial\vx}\right)^\mathbf{k} \left[\alpha_\mathbf{k}(\vx)p(\vx,t)\right],
\end{align}
where $\epsilon = 1/V$, $\vx = \mathbf{n}_V/ V$, $p(\vx, t) = VP(\mathbf{n}_V ,t)$, and $\mathbf{k} = (k_1, k_2, \cdots, k_q)$, $\sum_{\mathbf{k}} = \sum_{k_q} \cdots \sum_{k_2}\sum_{k_1}$, $\mathbf{k} ! = \prod_i (k_i !)$, and $\alpha_\mathbf{k}(\vx) = \sum_{\mathbf{r}}(\prod_i r_i^{k_i} )w(\vx, \mathbf{r}) $, $w(\vx, \mathbf{r})=W(\vX, \mathbf{r})/ V$. The solutions of the differential equation \eqref{KM-expan} with the infinite terms represents the exact time-dependent probability density of the scaled number of population $ \mathbf{n}_V/ V$. Then a natural question arises: could we obtain a corresponding diffusion process from this infinite order differential equation? The truth is that we can only get a ``local diffusion process approximation" for the scaled Markov jump process due to the following reason.

A common method to attack Eq. \eqref{KM-expan} is by truncating the higher order terms to the second order of $\epsilon$ to obtain a  Fokker-Planck equation (FPE). However, van Kampen \cite{van1992stochastic} pointed out that this method may fail if $\mathbf{n}_V(t)$ has large sizes of jumps. A concrete example was provided in the work \cite{vellela2009stochastic}: There exists an inconsistency between the stationary solutions of Kramers-Moyal FPE and the original master equation for the Schl\"ogl's model of a chemical reaction system which has bistable steady states.  The reason for the failure of Kramers-Moyal FPE is that we are only able to observe either the deterministic behavior of the process at the scale of the law of large number or the Gaussian fluctuations at the scale of the central limit theorem; however, there is no one scale to obtain both. To keep the first two order terms of the Kramers-Moyal expansion simultaneously to represent a diffusion process at a single scale is incorrect. Therefore, van Kampen \cite{van1992stochastic} suggested the $\Omega$ expansion which allows us to get a deterministic trajectory as $\epsilon \rightarrow 0$ and a local approximation near the deterministic trajectory at the scale $O(\sqrt{\epsilon})$ separately.

In the present work, we focus on the continuous representation of complex systems. We always start with random perturbations of dynamical systems represented by a sequence of stochastic differential equations (SDEs) parameterized by a small parameter $\epsilon$
\begin{align}  \label{macro.scale3}
    \rd \vX_\epsilon(t) = \vb(\vX_\epsilon) \rd t + [2\epsilon\mathbf{D}]^{\frac{1}{2}} \rd \mB(t),
\end{align}
where $\vX_\epsilon \in \mathbb{R}^n$, $\vb: \mathbb{R}^n \rightarrow \mathbb{R}^n$ stands for a drift function, the  $\mathbb{R}^n \times \mathbb{R}^n$ diffusion matrix $\mD$ is constant and positive semidefinite symmetric, and $\mB$ is the standard $n$-dimensional Brownian motion. This Langevin type equation is widely applicable for complex systems related to mechanics, and it gives us a clear picture of the entire dynamics including the drift and diffusion at one scale. Furthermore, by the rigorous mathematical theory of semigroup \cite{feller1954general}, every diffusion process represented by a SDE has a unique FPE to characterize the corresponding transition probability density $p_{\epsilon}(\vx, t)$
\begin{align}
\frac{\partial p_{\epsilon}}{\partial t} &= - \nabla \cdot \mathbf{J} [p_{\epsilon} ], \quad \mathbf{J}[p_{\epsilon}] \equiv \vb(\vx)p_{\epsilon} - \epsilon\mathbf{D}(\vx) \nabla p_{\epsilon}.
\label{seq-PDE1}
\end{align}
This line of reasoning to relate diffusion processes and FPEs has no ambiguity unlike the Kramers-Moyal FPEs.

\subsection{Random perturbations of diffusion processes} \label{ch1-2}

Our analysis of random perturbations of diffusion processes is by expansion in powers of $\epsilon$ for the sequence of SDEs \eqref{macro.scale3}, which follows the work of Freidlin and Wentzell  \cite{freidlin1998random}. As $\epsilon \rightarrow 0$, the sequence of SDEs converges to an ordinary differential equation (ODE) of the emergent deterministic trajectory by the Law of large number (LLN). To shift the sequence of SDEs to its deterministic trajectory with normalization by the scale $O(\sqrt{\epsilon})$, the rescaled sequence of SDEs converges to a time-inhomogeneous Gaussian process by the central limit theorem (CLT). For rare events in $O(1)$, they have the probability asymptotic to zero exponentially fast by the Large deviation principe (LDP). In comparison to Freidlin and Wentzell, there is another celebrated theory for the LDP by Donsker and Varadhan  \cite{donsker1975asymptotic1,donsker1975asymptotic2,donsker1976asymptotic3,donsker1983asymptotic4}. The main difference between them is that the Freidlin-Wentzell theory is about large deviations from a deterministic trajectory by small noise but the Donsker-Varadhan theory is regarding large deviations of certain process expectations for large time with the ergodic theorem.

In the present paper, we provide a trajectory-based proof for an emergent time-inhomogeneous Gaussian process in $\mathbb{R}^n$ near a deterministic trajectory under the CLT, which follows the proof for the particular case of $\mathbb{R}^1$ in the Freidlin-Wentzell's textbook \cite{freidlin1998random} (The idea of proof for $\mathbb{R}^n$  was suggested in the book but without details) and we further obtain a {\em Lyapunov differential equation} for the covariance of this Gaussian process. To the best of our knowledge, our derivation by the CLT for a small-noise process with respect to its deterministic trajectory is not included in the Donsker-Varadhan theory since the existence of deterministic trajectories is not required in their theory. In the field of statistical physics, this Lyapunov differential equation was mentioned in the works \cite{keizer2012statistical,van1992stochastic,gardiner2009stochastic}. However, all of those previous  works were based on the small noise expansion of the associated FPEs and each approach has some limitations: In \cite{van1992stochastic,gardiner2009stochastic}, the dynamics was restricted to one dimension; In \cite{keizer2012statistical}, the dynamics was for elementary processes in chemical reactions. Our approach for the Lyapunov differential equation is trajectory-based without transferring the original SDE problem to the problem of perturbations of partial differential equations (PDEs) and it is applicable for rather general multi-dimensional diffusion processes.

In contradistinction to the Freidlin-Wentzell theory and the Donsker-Varadhan theory, which are both from the standpoint of trajectories of systems, there is another approach of the LDP based on the PDEs: A logarithmic transformation to the differential
generator of diffusion processes was proposed by Fleming in 1978 \cite{fleming1977exit} then the PDE-based approach was applied to the LDP through solving the Hamilton-Jacobian equations (HJEs) by Evans and Ishii \cite{evans1985pde} and others. Feng and Kurtz \cite{feng2006large} generalized this approach by refining techniques on the viscosity solutions of HJEs so that the scope of applications of it is compatible with the Freidlin-Wentzell theory and the Donsker-Varadhan theory. This rigorous mathematical PDE-based approach is corresponding to the WKB method of solving FPEs, which was introduced early by theoretical physicists \cite{kubo1973fluctuation,graham1984weak}. In the present work, our analysis of stochastic limit cycles 
is carried out in parallel with both the small random perturbations of SDEs and the WKB approximation of the corresponding transition probability density, which can regarded as an example of a link between the trajectory-based and the PDE-based methods.  The
contradistinction provides a more comprehensive 
portrayal of the stochastic limit cycle.

\subsection{Time-inhomogeneous Gaussian processes from a transient state to an invariant set}

By relating those two methods, one of the important results obtained in this paper is the correspondence between the local Gaussian fluctuations along a deterministic path, limit cyle or not, and the curvarture of the leading order term in the WKB method near its infimum. 
In the early works, the connection between the CLT and the LDP of random processes can be found in the analysis of action functional for Gaussian random processes  \cite{freidlin1998random} and the LDP for the empirical measures of centered stationary Gaussian processes \cite{donsker1985large,bryc1997large}, in which the former follows the Friedlin-Wentzell theory and the later follows the Donsker-Varadhan theory. In the present paper, our work on the analysis of the CLT and the LDP of nonlinear systems with stochastic limit cycles, from a transient state to infinite time limit on an invariant set, are beyond those theories.

Globally, from the standpoint of probability, the existence of stationary distribution in the whole space for a stochastic stable limit cycle has been proved by Holland \cite{holland1978stochastically}. With the WKB method, characterizations of the stationary large deviation rate function near the cycle were studied in the previous work \cite{dykman1993stationary,vance1996fluctuations,ge2012landscapes,li2014landscape,lin2019quasi}. Locally, from the standpoint of trajectories, the dynamics is attracted to an invariant set but still capable of escaping from the set due to the multi-dimensional  fluctuations except the part tangential to the cycle. In the long run, the Gaussian fluctuations along the direction tangential to the cycle is eventually smeared out and the rest of fluctuations in the hyperplane perpendicular to the cycle are outward and damped out by the dissipation toward the limit cycle \cite{kurrer1991effect}.

In this paper, equipped with the Lyapunov equation for the covariance of the time-inhomogeneous Gaussian process, we characterize the fluctuations along the limit cycle by {\em asymptotic analysis}. Via a careful study of the interchange of limits of time $t \rightarrow \infty $ and $\epsilon \rightarrow 0$, with a coordinate transformation and dimension reduction on the cycle, we show that the Lyapunov equation becomes a $n-1 \times n-1$ {\em periodic Riccati differential equation} \cite{bittanti1984stability,pastor1993differential,chen2000existence,zhou2011periodic} and the solution of equation is a positive definite matrix. We further characterize the curvature of large deviation rate function on the limit cycle by the correspondence between the covariance matrix and the curvature established in our theory. 

The importance of stochastic limit cycle oscillations in physics was emphasized by Keizer \cite{keizer2012statistical} and van Kampen \cite{van1992stochastic}. In their books, specific examples with careful studies were provided but a general analysis was missing. Our analysis by both the trajectory-based and probability-based methods, from a transient state to an invariant set, helps us to paint a clear picture of dynamics in different scopes of space and time. Additionally, the present work can be regarded as an extension of the linear approximation theory of a stochastic nonlinear system with a fixed point as the steady state, which can be found in \cite{qian2001mathematical,kwon2005structure,qian2011nonlinear}, to an invariant set in $\mathbb{R}^n$.

\subsection{Organization of the paper}

In Sec. \ref{ch2}, we start with a rather general small-noise diffusion process represented by a sequence of non-linear multidimentional SDEs. Based on both the trajectory-based approach and the WKB method, key lemmas regarding the time-inhomogeneous Gaussain processes and a link to the large deviation rate function are provided.  In Sec. \ref{ch3}, we apply the lemmas to stochastic limit cycle oscillators. This approach is distinct to the previous works \cite{dykman1993stationary,vance1996fluctuations,ge2012landscapes}: (i) The works \cite{dykman1993stationary,vance1996fluctuations} are regarding fluctuations of limit cycles in chemical systems (The former \cite{dykman1993stationary} focused on analysis of a stationary FPE and the later \cite{vance1996fluctuations} applied the WKB method directly to a master equation.); (ii) The work \cite{ge2012landscapes} focused on the case of one-dimensional motion on a circle. In Sec. \ref{ch4}, we introduce the scaling hypothesis of diffusion processes to construct a sequence of dynamics parameterized by $\epsilon$. This scaling hypothesis not only serves as an useful mathematical tool of asymptotic analysis but also a scientific theory to justify the origin of $\epsilon$.

\section{Preliminaries} \label{ch2}

As we mentioned in Sec. \ref{ch1}, both discrete chemical kinetics and continuous mechanical motions successfully depict complex systems via introducing uncertainty. Based on the probability theory, the former is conventionally characterized by Markov jump processes (continuous-time and discrete-state) with the corresponding transition probability captured by  master equations, and the later is popularly described by diffusion processes (continuous-time and continuous-state) with SDEs. By introducing a parameter of the size of systems, at proper scales, both representations have their corresponding FPEs of the transition probability density and the Hamilton-Jacobi equations  (HJEs) of the large deviation rate function. In the present work, based on the continuous representation, we follow the direction SDE - FPE - HJE in a sequence. 

%In analysis, we rely on mathematical concepts of the law of large number (LLN), the central limit theorem (CLT), and the large deviation principle (LDP). 

%Based on these limit theorems, we have approximations of multidimensional nonlinear diffusion processes from two standpoints: (i) Random perturbations of the trajectory (Section \ref{ch2-2}). (ii) The WKB approximations of the probability density (Section \ref{ch2-3}). We will obtain four important preliminaries (Lemma \ref{thm:gaussian.z}, Lemma \ref{lemma-2}, Lemma \ref{lemma-3}, and Lemma \ref{thm:two-types}) for our main theorems of stochastic limit cycles in Section \ref{ch3}. 

%We emphasize that there exist essential differences between the derivations of those PDEs (FPEs and HJEs) of the two models. To clarify the differences, 

\subsection{Expansion in powers of a small parameter for diffusion processes}
\label{ch2-2}
Let us start from a sequence of SDEs defined in Eq. \eqref{macro.scale3}
\begin{align}  \label{macro.scale4}
    \rd \vX_\epsilon(t) = \vb(\vX_\epsilon) \rd t + [2\epsilon\mathbf{D}]^{\frac{1}{2}} \rd \mB(t),
\end{align}
and by the LLN, it converges to the following ordinary differential equation (ODE) as $\epsilon \rightarrow 0$
\begin{align} \label{det-dyn}
\rd \mathbf{x}(t) = \vb(\mathbf{x}) \rd t. 
\end{align}
Let $\hat{\mathbf{x}}(t)$ be the solution of this ODE with a given initial condition $\hat{\mathbf{x}}(0) = \hat{\mathbf{x}}_0$. 

We shall note that a direct application of small noise expansions for the process \eqref{macro.scale4} by $\vX_\epsilon(t) = \sum_{i=0}^n \epsilon^i \vX_i(t)$ may fail for certain types of drift functions $\vb$ \cite{gardiner2009stochastic}. Therefore, we need to expand $\vX_\epsilon(t)$ with a proper scale: By the scale of the CLT, we can define a random process $\vZ_{\epsilon}(t)$ near the deterministic trajectory  $\hat{\mathbf{x}}(t)$ 
\begin{align}
    \vZ_\epsilon(t) \equiv {\frac{\vX_{\epsilon}(t) - \hat{\mathbf{x}}(t)}{\sqrt{\epsilon}}}  
\label{cv}
\end{align}
and substitute Eq. (\ref{cv}) into Eq.  \eqref{macro.scale4}, we can derive that
\begin{equation}
\begin{alignedat}{2}
&&  \rd \vX_{\epsilon}(t) =  \rd \hat{\mathbf{x}}(t) + \sqrt{\epsilon} \rd \vZ_{\epsilon}(t) &=  \vb(\vX_{\epsilon}) \rd t + [2\epsilon\mathbf{D}]^{\frac{1}{2}} \rd \mB(t) + O(\epsilon)\\
\Rightarrow  && \quad \rd \vX_{\epsilon}(t) =  \rd \hat{\mathbf{x}}(t) + \sqrt{\epsilon} \rd \vZ_{\epsilon}(t) &=  \left( \vb(\hat{\vx}) + \sqrt{\epsilon} \mA(\hat{\vx})\vZ_\epsilon \right) \rd t + [2\epsilon\mathbf{D}]^{\frac{1}{2}} \rd \mB(t) + O(\epsilon) \\
\Rightarrow && \rd \vZ_{\epsilon}(t) &=   \mA(\hat{\mathbf{x}})  \vZ_{\epsilon} \rd t +  [2\mathbf{D}]^{\frac{1}{2}} \rd \mB(t) + O(\sqrt{\epsilon}) \label{cv-lin},
\end{alignedat}
\end{equation}
where $\mA(\hat{\vx}(t))$ is the Jacobian matrix of $\vb(\vx)$ evaluated at $\vx =\hat{
\vx}(t)$. We then follow the usual approach \cite{freidlin1998random} of perturbation theory to obtain an expansion in powers of the small parameter $\sqrt{\epsilon}$
\begin{align} \label{z-expansion}
\vZ_{\epsilon}(t) = \vZ(t) + \sqrt{\epsilon} \vZ^{(1)}(t) + \dotsi + {\sqrt{\epsilon}}^k \vZ^{(k)}(t) + \dotsi.
\end{align}
Apply the expansion of  $\vZ_\epsilon(t)$ in Eq. \eqref{z-expansion} to its SDE in Eq. \eqref{cv-lin}, we can obtain a SDE for the zeroth approximation of $\vZ_\epsilon(t)$
\begin{align}
\rd \vZ(t) &= \mA(\hat{\vx}) \vZ \rd t +  [2\mathbf{D}]^{\frac{1}{2}} \rd \mB(t).
\label{SDE-per}
\end{align}

The following lemma is about the solution of $\vZ(t)$:
\begin{lemma} \label{thm:gaussian.z}
If each element of the Jacobian matrix $\mA(\hat{\vx}(t))$ is continuous for all $t \geq 0$, then for every $t > 0$, $\vZ(t)$ is a Gaussian random variable $\vZ(t) \sim \mathcal{N}(\boldsymbol{\mu}(t), \mSigma(t)  )$  with
\begin{align}
\frac{ \rd \mmu(t)}{\rd t} &= \mA(\hat{\vx}) \mmu, \quad \quad \mmu(0) =  \hat{\mmu}_0, \label{dyn-cov0} \\
\frac{ \rd \mSigma(t)}{\rd t} &= \mA(\hat{\vx}) \mSigma + \mSigma   \mA(\hat{\vx})^T +  2\mD, \quad \quad \mSigma(0) = \hat{\mSigma}_0,
\label{dyn-cov1}
\end{align}
where $\mmu_0$ and $\mSigma_0 $ are given initial conditions. 
\end{lemma}
\begin{proof}
Under the assumption that each element of $\mA(\hat{\vx}(t))$ is continuous for all $t \geq 0$, there exits a fundamental matrix $\MM(t) \in \mathbb{R}^n \times \mathbb{R}^n$ satisfied the linear homogenous ordinary differential equation
\begin{align} \label{eqn:homogenous}
\rd\MM(t) = \mA(\hat{\vx}) \MM \rd t \quad \text{for all} \ t >0. 
\end{align}
Let $\vZ_0$ be the given initial condition for the dynamics \eqref{SDE-per}, we can verify the equation
\begin{align} \label{eqn:z}
    \vZ(t) = \MM(t) \left( \vZ_0 + \int_0^t \MM^{-1}(s)[2\mathbf{D}]^{\frac{1}{2}} \rd \mB(s)  \right) \quad \text{for all} \ t >0,
\end{align}
by differentiating the both sides of it with the It\^o lemma and Eq. \eqref{SDE-per} and Eq. \eqref{eqn:homogenous} as follows
\begin{align} 
     & \rd \left[ \MM(t) \left( \vZ_0 + \int_0^t \MM^{-1}(s)[2\mathbf{D}]^{\frac{1}{2}} \rd \mB(s)  \right)  \right] \notag \\
     &= \MM \rd \left( \vZ_0 + \int_0^t \MM^{-1}[2\mathbf{D}]^{\frac{1}{2}} \rd \mB  \right) + \rd \MM \left( \vZ_0 + \int_0^t \MM^{-1}[2\mathbf{D}]^{\frac{1}{2}} \rd \mB  \right) + \rd \MM \rd \left( \vZ_0 + \int_0^t \MM^{-1}[2\mathbf{D}]^{\frac{1}{2}} \rd \mB  \right) \notag \\
     &= [2\mathbf{D}]^{\frac{1}{2}} \rd \mB + \mA(\hat{\vx}) \MM \rd t  \MM^{-1} \vZ + \mA(\hat{\vx}) \MM \rd t \MM^{-1}  [2\mathbf{D}]^{\frac{1}{2}} \rd \mB \notag \\
    &= [2\mathbf{D}]^{\frac{1}{2}} \rd \mB + \mA(\hat{\vx}) \MM \MM^{-1} \vZ \rd t    \notag \\ 
    &= \rd \vZ(t). \notag
\end{align}

By Eq. \eqref{eqn:z}, for any constant vector $\va \in \mathbb{R}^n$, we have that
\begin{align}
    \va^T \vZ(t) &= \va^T \MM(t) \left( \vZ_0 + \int_0^t \MM^{-1}(s)[2\mathbf{D}]^{\frac{1}{2}} \rd \mB(s)  \right) = \sum_{i=1}^n \int_0^t f_i(s) \rd B_i(s),
\end{align}
where $B_i \in \mathbb{R}^1$ is a collection independent and identically distributed random variables from the standard Brownian motion $\mB = (B_1, B_2, \cdots, B_n)$, and $f_i:\mathbb{R}^1 \rightarrow \mathbb{R}^1$ is a collection of deterministic functions. Therefore,  $ \va^T \vZ(t) $ has to be a one-dimensional Gaussian random variable since it is a linear combination of a collection of independent one-dimensional  Gaussian random variables. Furthermore, by \cite{tong2012multivariate}, using the moment generating functions, arbitrary linear combinations of the random vector $\vZ(t)$ being an univariate Gaussian random variable implies that $\vZ(t)$ is a multivariate Gaussian random variable.

Next, we want to find expressions of the mean and the covariance of $\vZ(t)$ for every $t$. By the property of the standard Brownian motion, given a matrix $\mathbf{F}(t)$ which is independent of $\mB$, we have that
\begin{align} \label{prop.Brown}
   \mathbb{E}\left[\int_0^t \mathbf{F}(s) \rd \mB(s) \right] = \mathbf{0}, \quad \text{for all} \ t \geq 0.
\end{align}
With \eqref{eqn:z} and \eqref{prop.Brown}, the first and second moment of $\vZ(t)$ should satisfy
\begin{align} \label{eqn:gauss.mean.variance}
    \mathbb{E}[\vZ(t)] &= \MM(t) \vZ_0 \notag \\ 
    \mathbb{E}[\vZ(t) \vZ(t)^T] &= \mathbb{E}\left[\MM(t) \left( \vZ_0 + \int_0^t \MM^{-1}[2\mathbf{D}]^{\frac{1}{2}} \rd \mB  \right)  \left( \vZ_0 + \int_0^t \MM^{-1}[2\mathbf{D}]^{\frac{1}{2}} \rd \mB  \right)^T \MM(t)^T \right] \notag \\
    &= \MM(t)\vZ_0 \vZ_0^T \MM(t)^T + \MM(t)\mathbb{E}\left[\left( \int_0^t \MM^{-1}[2\mathbf{D}]^{\frac{1}{2}} \rd \mB  \right)  \left( \int_0^t \MM^{-1}[2\mathbf{D}]^{\frac{1}{2}} \rd \mB  \right)^T  \right] \MM(t)^T \notag \\
    &= \MM(t)\vZ_0 \vZ_0^T \MM(t)^T  + 2\MM(t) \left( \int_0^t \MM^{-1}(s)\mathbf{D} \MM^{-T}(s) \rd s \right) \MM(t)^T, 
\end{align}
where we applied the It\^o isometry to the last equation. Since $\mmu(t) = \mathbb{E}[\vZ(t)] = \MM(t) \vZ_0$, and $\mSigma(t) =\mathbb{E}[\vZ(t) \vZ(t)^T] - \mathbb{E}[\vZ(t)] \mathbb{E}[ \vZ(t)^T] = 2\MM(t)\left( \int_0^t \MM^{-1}(s)\mathbf{D} \MM^{-T}(s) \rd s \right) \MM(t)^T $, by taking derivatives of them with respect to time, we thus obtain dynamics of $\mmu(t)$ and $\mSigma(t)$ as follows
\begin{align}
\frac{ \rd \mmu(t)}{\rd t} &= \mA(\hat{\vx}) \mmu, \quad \quad \mmu(0) =  \hat{\mmu}_0, \\
\frac{ \rd \mSigma(t)}{\rd t} &= \mA(\hat{\vx}) \mSigma + \mSigma   \mA(\hat{\vx})^T +  2\mD, \quad \quad \mSigma(0) = \hat{\mSigma}_0,
\label{dyn-cov2}
\end{align}
where $\hat{\mmu}_0$ and $\hat{\mSigma}_0 $ are given initial conditions.

\end{proof}

By Lemma \ref{thm:gaussian.z}, we obtained a time-inhomogeneous Gaussian process from a   multi-dimensional nonlinear diffusion process at the scale of the CLT with the covariance captured by  the Lyapunov differential equation \eqref{dyn-cov2}. Additionally, this lemma can be applied to solve the FPE
\begin{align} \label{seq-PDE1.5}
\frac{\partial p_{\epsilon}}{\partial t} &= - \nabla \cdot \mathbf{J} [p_{\epsilon} ], \quad \mathbf{J}[p_{\epsilon}] \equiv \vb(\vx)p_{\epsilon} - \epsilon\mathbf{D} \nabla p_{\epsilon}
\end{align}
with certain boundary conditions. Since the function $\vb$ is nonlinear and multi-dimensional, this type of PDE problems may not be easy to solve directly. Under certain conditions \cite{freidlin1998random,jiang2004mathematical}, the diffusion process \eqref{macro.scale4} is associated with this FPE. By expanding the solution of the FPE \eqref{seq-PDE1.5} and the diffusion process \eqref{macro.scale4} respectively
\begin{align}
    &p_\epsilon(\vx, t) = \frac{1}{\sqrt{\epsilon}} \hat{p}_\epsilon(\vz, t) \quad \text{and} \quad  \hat{p}_\epsilon(\mathbf{z},t) = \hat{p}_0(\mathbf{z}, t) + \sum_{n=1}^\infty (\sqrt{\epsilon})^n \hat{p}_n(\mathbf{z}, t) , \\
    & \vZ_\epsilon(t) = {\frac{\vX_{\epsilon}(t) - \hat{\mathbf{x}}(t)}{\sqrt{\epsilon}}} \quad \text{and} \quad \vZ_\epsilon(t) = \vZ(t) + \sum_{n=1}^{\infty} (\sqrt{\epsilon})^n\vZ^{(n)}(t) , \label{prob.v.s.trajectory}
\end{align}
we can check that $\hat{p}_0(\mathbf{z}, t)$ is the probability density of $\vZ(t)$. Following from Lemma \ref{thm:gaussian.z}, we thus obtain an approximation solution of the FPE by having the dynamics of the mean and covariance of $\vZ(t)$. To transform the FPE problem into the problem of a diffusion process, the above example is an application of Lemma \ref{thm:gaussian.z} to attack a complicated boundary value problem.  

\subsection{Approximations by the asymptotic theory akin to the WKB}
\label{ch2-3}

In Sec. \ref{ch1-1}, we pointed out that using the Kramers-Moyal Fokker-Planck equation for a mater equation may fail in some cases. Instead, the deterministic behavior and the local fluctuation of a jump Markov process can be obtained by the $\Omega$ expansion with respect to two different scales. In addition to the $\Omega$ expansion, another approach by the WKB approximation has been applied to give a full analysis of master equations \cite{kubo1973fluctuation,gang1987stationary,vance1996fluctuations,nakanishi2013hamilton}. In this method, the WKB ansatz is assumed for the solution of the Kramers-Moyal expansion of a master equation without truncating higher-order terms, so this method has no problem unlike the Kramers-Moyal Fokker-Planck equation. As the success of the WKB approximation of master equations for Markov jump processes, this method has also been used to the associated Fokker-Planck equations of diffusion processes \cite{graham1984weak,Qian-2017}. Here we want to link our trajectory-based approach in Sec. \ref{ch2-2} to the probability-based approach by the WKB approximation.

Recall that the path behaviors of the diffusion process is described by the n-dimensional SDE \eqref{macro.scale4}
\begin{align} \label{pert-sys}
\rd\vX_{\epsilon}(t) = {\mathbf{b}}(\vX_{\epsilon}) \rd t + \sqrt{2\epsilon \mD} \rd \mB_t.
\end{align}
In order to link the SDE \eqref{pert-sys} to the WKB ansatz, we need to find its probability-density representation by a FPE. From Eq. \eqref{seq-PDE1.5} to Eq. \eqref{prob.v.s.trajectory}, we have illustrated a way to convert a PDE problem to a SDE problem; on the other hand, by the semigroup approaches \cite{feller1954general,jiang2004mathematical}, we can also convert a SDE problem to a PDE problem. Under certain conditions \cite{jiang2004mathematical}, the original SDE problem can be characterized by the solution of the FPE
\begin{align}
\frac{\partial p_{\epsilon}}{\partial t} &= - \nabla \cdot \mathbf{J} [p_{\epsilon} ], \quad \mathbf{J}[p_{\epsilon}] \equiv \vb(\vx)p_{\epsilon} - \epsilon\mathbf{D} \nabla p_{\epsilon}.
\label{seq-PDE2}
\end{align}
We shall note that, as $\epsilon \rightarrow 0$, the FPE reduces to a first-order differential equation so the perturbation of the solution $p_\epsilon$ follows the singular perturbation theory. To attack this singular perturbation problem, we adopt an asymptotic series of $p_\epsilon$  with a proper scaled variable $\mathbf{z} =  (\vx - \hat{\vx}) / \sqrt{\epsilon} $,
\begin{align} \label{asymptotic1}
    p_\epsilon(\vx, t) = \frac{1}{\sqrt{\epsilon}} \hat{p}_\epsilon(\vz, t) \quad \text{and} \quad  \hat{p}_\epsilon(\mathbf{z},t) = \sum_{n=0}^\infty (\sqrt{\epsilon})^n\hat{p}_n(\mathbf{z}, t) .
\end{align}
In parallel, there is another complete asymptotic theory for the solution of FPE akin to the WKB ansatz \cite{kubo1973fluctuation,graham1984weak}
\begin{align} \label{asymptotic2}
    p_\epsilon(\vx, t)= a(\epsilon, t) \exp \left[ -\frac{1}{\epsilon} \sum_{n=0}^\infty \phi_n(\vx, t)\epsilon^n  \right], 
\end{align} 
where $ a(\epsilon, t)$ is a normalization factor. The expansion \eqref{asymptotic2} with the series $\epsilon^{-1}\phi_0 + \phi_1 + \epsilon\phi_2 + \cdots$ was justified by the {\em extensive property} of $p_\epsilon(\vx ,t)$, i.e. it keeps the form \eqref{asymptotic2} as time evolves \cite{kubo1973fluctuation}. 

We will connect those two types of expansions in Lemma \ref{thm:two-types}. To prove the lemma, we first give two useful lemmas on the asymptotic evaluation of various integrals \cite{bender2013advanced}. All the proofs can be found in Appendix \ref{proof:lemmas}. 

\begin{lemma} \label{lemma-2}
For sufficiently smooth scalar functions
$f(\vx)$ and $h(\vx)$, $\vx\in\mathbb{R}^n$,

\begin{equation}
       \int_{\mathbb{R}^n}
                     f(\vx) e^{-h(\vx)/\epsilon}\rd\vx
  = \sqrt{\frac{2\pi\epsilon}{\det[\nabla\nabla h(\vx^*)] }}  e^{-\frac{h(\vx^*)}{\epsilon}}
       \Big[ f(\vx^*) + \epsilon \eta(\vx^*) 
         + O\big(\epsilon^2\big) \Big] \label{eq-13-a} ,
\end{equation}
\begin{eqnarray}
       && \frac{\displaystyle \int_{\mathbb{R}^n}
                     f(\vx) e^{-h(\vx)/\epsilon}\rd\vx }{\displaystyle \int_{\mathbb{R}^n} e^{-h(\vx)/\epsilon}\rd\vx }  \label{eq-13-b}
\\
  &=&  f(\vx^*) + \epsilon \left[\frac{f''_{ij}(\vx^*)\Xi_{ij}}{2}-
     \frac{f'_i(\vx^*)h'''_{jk\ell}(\vx^*)\Xi^{\frac{1}{2}}_{i\mu} 
       \Xi^{\frac{1}{2}}_{j\nu}\Xi^{\frac{1}{2}}_{k\rho}
            \Xi^{\frac{1}{2}}_{\ell\kappa}\Theta_{\mu\nu\rho\kappa}}{6}  \right]  +
         O(\epsilon^2),
\nonumber
\end{eqnarray}
as $\epsilon\to 0$, in which Einstein's summation rule is 
adopted, $\vx^*$ is the global minimum 
of $h(\vx)$, and
\begin{eqnarray}
   \eta(\vx^*) &=&  \frac{f''_{ij}(\vx^*)\Xi_{ij}}{2}-\left[
     \frac{f'_i(\vx^*)h'''_{jk\ell}(\vx^*)}{6}+
            \frac{f(\vx^*)h''''_{ijk\ell}(\vx^*)}{24}\right]
    \Xi^{\frac{1}{2}}_{i\mu} 
       \Xi^{\frac{1}{2}}_{j\nu}\Xi^{\frac{1}{2}}_{k\rho}
            \Xi^{\frac{1}{2}}_{\ell\kappa}\Theta_{\mu\nu\rho\kappa}
\nonumber\\
   &+& \frac{f(\vx^*)[h'''_{ijk}(\vx^*)]^2}{72} \Xi^{-\frac{1}{2}}_{i\mu}
         \Xi^{-\frac{1}{2}}_{i\mu'}\Xi^{-\frac{1}{2}}_{j\nu}\Xi^{-\frac{1}{2}}_{j\nu'}\Xi^{-\frac{1}{2}}_{k\rho}\Xi^{-\frac{1}{2}}_{k\rho'} 
        \Lambda_{\mu\mu'\nu\nu'\rho\rho'}. \label{eq-13-c}
\end{eqnarray}
The covariance matrix $\mXi = \big[\nabla\nabla h(\vx^*)\big]^{-1}$, 
and the multi-indexed $\Theta_{ijk\ell}$ and
$\Lambda_{\mu\mu'\nu\nu'\rho\rho'}$ are
\begin{eqnarray}
 \Theta_{\mu\nu\rho\kappa} &=& \int_{\mathbb{R}^n}
  \frac{ y_{\mu}y_{\nu}y_{\rho}y_{\kappa}
           }{\big(2\pi\big)^{n/2} }
   \exp\left[-\frac{\vy^T\vy}{2}\right] \rd\vy, 
\label{Theta}\\
  \Lambda_{\mu\mu'\nu\nu'\rho\rho'} &=&
   \int_{\mathbb{R}^n}
  \frac{ y_{\mu}y_{\mu'}y_{\nu}y_{\nu'}y_{\rho}y_{\rho'}
           }{\big(2\pi\big)^{n/2} }
   \exp\left[-\frac{\vy^T\vy}{2}\right] \rd\vy.
\end{eqnarray}

\end{lemma}

By Lemma \ref{lemma-2}, we can obtain the following lemma, which is very useful in the integrals with respect to a probability density approximated by the WKB method.

\begin{lemma}  \label{lemma-3} For sufficiently smooth functions
$f(\vx)$, $g(\vx)$, and $h(\vx)$, $\vx\in\mathbb{R}^n$,
\begin{eqnarray}
  && \quad \frac{\displaystyle \int_{\mathbb{R}^n}
                     f(\vx)g(\vx)e^{-\frac{h(\vx)}{\epsilon}}\rd \vx}{
            \displaystyle  \int_{\mathbb{R}^n}
  g(\vx) e^{-\frac{h(x)}{\epsilon}}\rd \vx }
\\
  &=&   f(\vx^*) + \epsilon \left[ f'_i(\vx^*)(\log g)'_i(\vx^*) \Xi_{ij}+ \frac{f''_{ij}(\vx^*)\Xi_{ij}}{2}-
     \frac{f'_i(\vx^*)h'''_{jk\ell}(\vx^*)\Xi^{\frac{1}{2}}_{i\mu} 
       \Xi^{\frac{1}{2}}_{j\nu}\Xi^{\frac{1}{2}}_{k\rho}
            \Xi^{\frac{1}{2}}_{\ell\kappa}\Theta_{\mu\nu\rho\kappa}}{6}  \right]  +
         O(\epsilon^2), \notag 
\label{eq-14}
\end{eqnarray}
as $\epsilon\to 0$, in which Einstein's summation rule is 
adopted, $\vx^*$ is the global minimum 
of $h(\vx)$,

\end{lemma}

Now, we are ready to relate the two kinds of asymptotic series. Recall that the two expansions are
\begin{align} 
     \hat{p}_\epsilon(\vz, t) &= \sum_{n=0}^\infty (\sqrt{\epsilon})^n\hat{p}_n(\mathbf{z}, t),  \label{two-expansions1}\\
     p_\epsilon(\vx, t) &= a(\epsilon, t) \exp \left[ -\frac{1}{\epsilon} \sum_{n=0}^\infty \phi_n(\vx, t)\epsilon^n  \right], \label{two-expansions2}
\end{align}
where $\vz = (\vx - \hat{\vx})$ and $ p_\epsilon(\vx, t) = \hat{p}_\epsilon(\vz, t) /\sqrt{\epsilon}$. Since we will focus on the first two orders in the WKB approximation, we specifically denote that $\phi_0 := \varphi$ and $\phi_1 := \ln \omega$. These two functions have particular meanings: $\varphi$ is known as the large deviation rate function \cite{freidlin1998random,zeitouni1998large}, and $\omega$ is known as the prefactor \cite{graham1984weak,ge2012landscapes},  or the phase space factor \cite{van1992stochastic}, or degeneracy in the classical statistical mechanical terminology  \cite{Qian-2017}.

Recall that the time-dependent matrix $\mSigma(t)$ in Lemma \ref{thm:gaussian.z} is the covariance matrix of $\mathbf{Z}(t)$ and we have checked that $\mathbf{Z}(t)$ has the density $\hat{p}_0(\vz, t)$. Therefore, for  $\mSigma(t)$, it has the formula $\Sigma_{ij}(t) = \int_{\mathbb{R}^n} z_i z_j \hat{p}_0(\vz, t)\rd \vz $, for $1\leq i \leq n, \ 1\leq i \leq n$. Here we further define a time dependent first moment vector $\mathbf{m}(t)$ with respect to $\hat{p}_1(\vx, t)$ by ${m}_i(t) = \int_{\mathbb{R}^n} z_i \hat{p}_1(\vz, t)\rd \vz $, for $1\leq i \leq n$. Note that the functions $\hat{p}_0$ and $\hat{p}_1$ are given in the expansion \eqref{two-expansions1}. Under this framework,  $\mSigma(t)$ and $\mathbf{m}(t)$ must satisfy the differential equations
\begin{align} 
    \frac{ \rd \mSigma(t)}{\rd t} &= \mA(\hat{\vx}(t)) \mSigma + \mSigma   \mA(\hat{\vx}(t))^T +  2\mD, \quad \quad \mSigma(0) = \hat{\mSigma}_0, \label{dynamics.of.mSigam-ma} \\ 
     \frac{\rd \mathbf{m}(t) }{\rd t} &= \mA(\hat{\vx}(t)) + \mathbf{H}(\hat{\vx}(t))\mathbf{\Sigma}, \quad \mathbf{m}(0)= \hat{\mathbf{m}}_0 \label{dynamics.of.mSigam-mb},
\end{align}
in which the initial conditions $\hat{\mSigma}_0$ and $\hat{\mathbf{m}}_0$ are given by the distribution of $\vX_\epsilon(0)$. For example, if the initial probability density of $\vX_\epsilon(0)$ is purely Gaussian, then $\hat{p}_1(\vz, 0) = 0$ for all $\vz$ hence $\hat{\mathbf{m}}_0 = 0$. Furthermore, $\mathbf{A}(\vx)$ is the Jacobian matrix of $\vb(\vx)$, and $\mathbf{H}(\vx)$ is a rank 3 tensor with $\mathbf{H}_i(\vx)$ being the Hessian matrix of $b_i(\vx)$, $1\leq i \leq n$. Eq. \eqref{dynamics.of.mSigam-ma} of $\mSigma(t)$ is from Lemma \ref{thm:gaussian.z}, and Eq. \eqref{dynamics.of.mSigam-mb} can be verified by plugging the expansion \eqref{two-expansions1} into the the Fokker-Planck equation \eqref{seq-PDE2} and using integration by parts.

Based on the above setup,  $\mSigma(t)$,  $\mathbf{m}(t)$ are related to the expansion \eqref{two-expansions1} and $\varphi(\vx, t)$, $\omega(\vx, t)$ are related to the expansion \eqref{two-expansions2}, then we have the following lemma for their correspondence. Recall that  $\hat{\vx}(t)$ is the emergent deterministic trajectory of the diffusion process $\vX_\epsilon(t)$ as $\epsilon \rightarrow 0$.

\begin{lemma}  \label{thm:two-types} 
$\mSigma(t),$  $\mathbf{m}(t)$, $ \varphi(\hat{\vx}(t), t)$, and $\omega(\hat{\vx}(t), t)$ must satisfy the equations
\begin{align} 
    \mSigma(t)  &=  \left[\nabla\nabla \varphi(\hat{\vx}(t), t)\right]^{-1}, \label{sigma-WKB} \\ 
  \mathbf{m}(t) &=  \left[\nabla\nabla \varphi(\hat{\vx}(t), t)\right]^{-1} \nabla \log\omega(\hat{\vx}(t), t)  - \frac{1}{6} \nabla\nabla\nabla \varphi(\hat{\vx}(t), t) \Theta,  \label{m-WKB}
\end{align}
where $
    \big( \nabla\nabla\nabla \varphi(\vx, t) \Theta \big)_i = \sum_{jkl\mu\nu\kappa\rho} \varphi'''_{jk\ell}(\vx, t)\Xi^{\frac{1}{2}}_{i\mu} 
       \Xi^{\frac{1}{2}}_{j\nu}\Xi^{\frac{1}{2}}_{k\rho}
            \Xi^{\frac{1}{2}}_{\ell\kappa}\Theta_{\mu\nu\rho\kappa}$, $ \mXi = \left[ \nabla\nabla\varphi(\vx, t)\right]^{-1}$, and $\Theta$ is defined by Eq. \eqref{Theta} in Lemma \ref{lemma-2}. 
\end{lemma}

\begin{proof}
By the change of variable $\vz = (\vx - \hat{\vx})/\sqrt{\epsilon}$, 
we have the following two equations
\begin{align} 
 \int (\sqrt{\epsilon}\mathbf{z}) (\sqrt{\epsilon}\mathbf{z})^T \hat{p}_{\epsilon}(\mathbf{z},t) \rd\mathbf{z} &= \int (\mathbf{x}- \hat{\mathbf{x}}) (\mathbf{x}- \hat{\mathbf{x}})^T p_{\epsilon}(\mathbf{x},t) \rd\mathbf{x}, \label{expand-1} \\
 \int (\sqrt{\epsilon}\mathbf{z}) \hat{p}_{\epsilon}(\mathbf{z},t) \rd\mathbf{z} &= \int (\mathbf{x}- \hat{\mathbf{x}}) p_{\epsilon}(\mathbf{x},t) \rd\mathbf{x} \label{expand-2}.
\end{align}
Plug the expansion \eqref{two-expansions1} into the left side of \eqref{expand-1}, by Lemma \ref{thm:gaussian.z}, the left side of \eqref{expand-1} becomes
\begin{align} \label{epsilon-mSigma}
    \epsilon\mSigma(t)+ o(\epsilon).
\end{align}
For a fixed $t$, the point $\vx = \hat{\vx}(t)$ on the deterministic trajectory is the global minimum of $\varphi(\vx, t)$. Therefore, to plug the expansion \eqref{two-expansions2} into the right side of \eqref{expand-2}, by Lemma \ref{lemma-3}, we have
\begin{align}  \label{epsilon-varphi}
     \epsilon \left[\nabla\nabla\varphi(\hat{\vx}(t), t)\right]^{-1} + o(\epsilon).
\end{align}
With Eq. \eqref{epsilon-mSigma} and Eq. \eqref{epsilon-varphi}, we thus obtain 
\begin{align} \label{eqn:d.sigma}
   \mSigma(t)  =  \left[\nabla\nabla\varphi(\hat{\vx}(t), t)\right]^{-1}. 
\end{align}
By a similar approach, applying Lemma \ref{thm:gaussian.z} to the left side of \eqref{expand-2} and Lemma \ref{lemma-3} to the right side of it, we have that
\begin{align}  \label{eqn:am}
    \mathbf{m}(t) =  \left[\nabla\nabla\varphi(\hat{\vx}(t), t)\right]^{-1} \nabla \ln\omega(\hat{\vx}(t), t)  - \frac{1}{6} \nabla\nabla\nabla \varphi(\hat{\vx}(t), t) \Theta.
\end{align}

\end{proof}

The leading order $\varphi(\vx,t)$ of the time-dependent WBK ansatz \eqref{asymptotic2} is known as a time-dependent large deviation rate function. Our work (Lemma \ref{thm:two-types})  relates the curvature of the time-dependent large deviation rate function near its infimum with the local Gaussian fluctuations of diffusion processes. In the case of independent and identically distributed (i.i.d.) random variables sampling, the inverse of the curvature of large deviation rate function equivalent to the variance of each random variable is one of the important properties of the rate function \cite{bryc1993remark,touchette2009large}. Eq. \eqref{sigma-WKB} in Lemma \ref{thm:two-types} can be regarded as an extension of this property to the case of random processes. The Freidlin-Wentzell theory \cite{freidlin1998random} gives a clear definition of the large deviation rate function of random processes. From the trajectory standpoint, the action functional is defined as \cite{freidlin1998random}
\begin{align} \label{action}
    \mathcal{S}_{0,t}(\xi)=\frac{1}{4}\int_{0}^{t}[\dot{\xi_{s}}-\vb(\xi_{s})]\mD^{-1}[\dot{\xi_{s}}-\vb(\xi_{s})]\rd s,
\end{align}
where $\xi$ is the set of all smooth paths of the process \eqref{macro.scale4} on the interval $[0,t]$. Then the time-dependent large deviation rate function is the minimum of action among the set of $\xi$ 
\begin{align} \label{min-action}
    \varphi(\vx, t) = \min_{\xi_0 = \vx_0, \xi_t=\vx} \mathcal{S}_{0,t}(\xi),
\end{align}
in which $\vx_0$ and $\vx$ are the initial and end conditions of the process, respectively.  

By Eqs. \eqref{action} and \eqref{min-action},  $\varphi(\vx,t)$ is no longer just a mathematical concept of the leading order term of the logarithmic asymptotics of probability densities. Borrowing the concept from classical mechanics, the integrand in the integral \eqref{action} is called the Lagrangian of the action and there is a corresponding Hamiltonian of the system defined by the Legendre dual of the Lagrangian \cite{zhou2016construction}
\begin{align} \label{Hamiltonian}
 H(\xi, \mathbf{p}  ) = \vb(\xi) \cdot \mathbf{p} + \mD \mathbf{p} \cdot \mathbf{p}  .        
\end{align}
Furthermore, based on the Hamiltonian given in Eq. \eqref{Hamiltonian}, the large deviation rate function has to satisfy the Hamilton-Jacobi equation
\begin{align} \label{eq:HJE1}
    \frac{\partial \varphi(\vx , t)}{\partial t} = - H(\vx, \nabla \varphi). 
\end{align}
Finding solutions of the Hamilton-Jacobi equation \eqref{eq:HJE1} is still an open problem. By Lemma \ref{thm:two-types}, with the dynamics of $\mSigma(t)$  in Eq. \eqref{dynamics.of.mSigam-ma} and $\mathbf{m}(t)$  in Eq. \eqref{dynamics.of.mSigam-mb}, if the solution of the prefactor $\omega(\vx, t)$ is given, e.g., an uniform prefactor, then we can derive the dynamics of $\nabla \nabla\varphi(\hat{\vx}(t),t)$ and $\nabla\nabla\nabla\varphi(\hat{\vx}(t),t)$. These results can help us approximate the solution of the HJE near its infimum: by the multivariable Taylor's expansion, for $\| \vx - \hat{\vx} \| < \delta$, we have a {\em third order approximation}
\begin{align} \label{HJE-approx}
    \varphi(\vx , t) = \frac{1}{2}\left[ (\vx - \hat{\vx}(t)) \cdot \nabla\right]^2 \varphi(\hat{\vx}(t) ,t) + \frac{1}{3}\left[ (\vx - \hat{\vx}(t)) \cdot \nabla\right]^3 \varphi(\hat{\vx}(t) ,t)   + o(\delta^3),
\end{align}
and the first two terms on the right side can be numerically solved with the dynamics of $\nabla \nabla\varphi(\hat{\vx}(t),t)$ and $\nabla\nabla\nabla\varphi(\hat{\vx}(t),t)$ obtained by our theory.

To summarize the novelty and the significance of Lemma \ref{thm:two-types} in the following three points: 
\begin{enumerate}
    \item We show rigorously that the covariance matrix of the local time-inhomogeneous Gaussian process near a deterministic trajectory is equivalent to the inverse of the curvature of the time-dependent large deviation rate function near its infimum. 
    \item By having the dynamics of $\nabla \nabla\varphi(\hat{\vx}(t),t)$ and $\nabla\nabla\nabla\varphi(\hat{\vx}(t),t)$ from the lemma, we can obtain a third-order approximation of the solution of the HJE \eqref{eq:HJE1} near its infimum. 
    \item For analyzing a stochastic stable limit cycle in the next section, with the Lyapunov differential equation of  $\mSigma(t)$  \eqref{dynamics.of.mSigam-ma} and the equation $\mSigma(t) = \left[\nabla\nabla \varphi(\hat{\vx}(t), t)\right]^{-1}$ in the lemma, we can further study the asymptotic behaviors of the time-inhomogeneious Gaussian process from a transient state to an invariant set, and relate this result to the previous works \cite{vance1996fluctuations,ge2012landscapes,li2014landscape,lin2019quasi} on the curvature of the stationary large deviation rate function near a limit cycle.
\end{enumerate}

\section{Main results of stochastic limit-cycle oscillations}
\label{ch3}

In this section, we focus on nonlinear stochastic complex systems having stable limit cycles at the macroscopic scale. In chapter XIII. 7 of the textbook by van Kampen \cite{van1992stochastic}, the author proposed two examples of stochastic systems with stable limit cycles: one is dynamics of the Brusselator in the chemical reaction and the other one is the generalized Ginzburg–Landau equation in statistical mechanics. The model of the former started from a master equation for the Markov jump process and the later began with a SDE for the diffusion process. More examples of stochastic chemical kinetics with limit cycle oscillators were thoroughly discussed in the previous works \cite{dykman1993stationary,vance1996fluctuations}. In contradistinction to stochastic chemical kinetics, our work follows the idea of the second example in the van Kampen's book along the line of the continuous representation of complex systems: we start with a randomly perturbed diffusion process satisfied the sequence of SDEs \eqref{macro.scale4}. Recall it has the form 
\begin{align}  \label{macro.scale5}
    \rd \vX_\epsilon(t) = \vb(\vX_\epsilon) \rd t + [2\epsilon\mathbf{D}]^{\frac{1}{2}} \rd \mB(t).
\end{align}
In this section, we assume $\mD$ is positive definite in particular. Furthermore, there is an emergent deterministic dynamics as $\epsilon \rightarrow 0$,
\begin{align} 
\rd \mathbf{x}(t) = \mathbf{b}(\mathbf{x}) \rd t, \label{ODE}
\end{align}
and the solution  $\hat{\mathbf{x}}(t)$ of this ODE \eqref{ODE} with initial condition $\hat{\mathbf{x}}(0)$ has a invariant solution as a stable limit cycle $\Gamma$. Recall that the corresponding FPE of Eq. \eqref{macro.scale5} is
\begin{align}
\frac{\partial p_{\epsilon}}{\partial t} &= - \nabla \cdot \mathbf{J} [p_{\epsilon} ], \quad \mathbf{J}[p_{\epsilon}] \equiv \vb(\vx)p_{\epsilon} - \epsilon\mathbf{D} \nabla p_{\epsilon}.
\label{seq-PDE3}
\end{align}
Under this  framework, we will apply the preliminaries obtained in Sec. \ref{ch2} to analyze stochastic limit oscillators by taking time goes to infinity. The following section is about asymptotic analysis of this transition from a finite time to the infinite time limit.

%In this whole section \ref{ch3}, we will apply the preliminaries obtained in Section \ref{ch2} to the random perturbed dynamical system having a limit cycle as an invariant set. The emergent deterministic dynamics \eqref{ODE} of the system has to be nonlinear to have the phenomenon of limit cycle oscillations. By the preliminaries obtained in Section \ref{ch2}, we will find the linear approximations of the nonlinear drift term near its deterministic trajectory $\hat{\vx}(t)$, then take time goes to infinity for our analysis of the limit cycle. The following section \ref{ch3-1} is about asymptotic analysis of this transition from a finite time to the infinite time limit. 

\subsection{The process asymptotic to a time-invariant set} \label{ch3-1}
In Sec. \ref{ch2-3}, we relate the time-dependent large deviation rate function and the prefactor in the WKB method with the stochastic trajectories of randomly perturbed dynamics systems. By this correspondence, we further inspect asymptotic behaviors of the time-dependent large deviation rate function and the prefactor as time goes to infinity. By plugging the WKB ansatz \eqref{asymptotic2} into the FPE \eqref{seq-PDE2}  with equating likeorder terms, we obtain two partial differential equations of $\varphi(\vx.t)$ and $\omega(\vx,t)$ respectively,
\begin{align}
    \frac{\partial\varphi(\vx, t)}{\partial t} &= - \nabla\varphi(\vx, t) \cdot \vgamma(\vx, t)  \label{eq:rate-function}\\
    \frac{\partial\omega(\vx, t)}{\partial t} &= -\nabla\cdot(\vgamma(\vx,t) \omega(\vx,t)) - \mD \nabla \varphi(\vx,t) \cdot \nabla \omega(\vx, t), \label{eq:prefactor}
\end{align}
where $\vgamma(\vx,t) = \mD \nabla \varphi(\vx, t) + \vb(\vx)$. We will later discuss the meaning of $\vgamma(\vx,t)$, which is closely related the concept of probability flux and the Onsager's thermodynamic force.

Let us start from analysis of the solution of $\varphi(\vx, t)$ in Eq. \eqref{eq:rate-function}, which is a Hamilton-Jacobi equation. The HJE derived in this place by  equating the leading order term in the time-dependent WKB ansatz is the same one obtained in the previous derivation (Eq. \eqref{min-action} - Eq. \eqref{eq:HJE1}) by minimizing of action among all possible smooth paths.  However, when we further analyze time-invariant solutions of this HJE, we shall notice an essential difference between the WKB-type method and the trajectory-based method, which is due to the interchange of limits of $t$ and $\epsilon$.

Following the idea of WKB anstaz \eqref{asymptotic2} for the time-dependent probability density, if the invariant probability exists, it could be written as the asymptotic form
\begin{align} \label{wkb:stationary}
    \pi_\epsilon(\vx) = \hat{a}(\epsilon) \exp \left[ - \frac{1}{\epsilon} \hat{\varphi}(\vx) +  \ln \omega(\vx) + O(\epsilon)  \right]
\end{align}
and it is equivalent to say that 
\begin{align}
    \hat{\varphi}(\vx) = - \lim_{\epsilon \rightarrow 0} \lim_{t \rightarrow \infty} \ln p_{\epsilon}(\vx ,t|\vx_{0},0),
\end{align}
where $\hat{\varphi}$ is independent of the initial condition $\vx_0$ and is well-defined in the whole space $\mathbb{R}^n$. On the other hand, from the standpoints of trajectories, the quasipotential of the system is defined by 
\begin{align}
\varphi(\vx;\vx_f):=\inf_{t>0}\inf_{\xi}\{\mathcal{S}_{0,t}(\xi):\xi_{0}=\vx_f,\xi_{t}=\vx\},
\label{quasipotential}
\end{align}
where $\vx_f$ is a fixed point of the deterministic dynamics $\hat{\vx}'=\vb(\hat{\vx})$. This definition extends the definition of minimizing the action \eqref{min-action} from a fixed $t$ to all $t > 0$ and it has a corresponding probabilistic representation 
\begin{align}
    \varphi(\vx; \vx_f) = -  \lim_{t \rightarrow \infty} \lim_{\epsilon \rightarrow 0} \ln p_{\epsilon}(\vx ,t |\vx_f,0).
\end{align}

The quasipotential \eqref{quasipotential} is one of the invariant solutions of the HJE \eqref{eq:rate-function} \cite{freidlin1998random}. It may contain non-differential parts since the HJE can have a non-smooth solution after a certain finite time by studying the characteristics of it \cite{evans1998partial}. We shall emphasize that the two potentials, $ \hat{\varphi}(\vx)$ and $\varphi(\vx;\vx_f)$, have the same shape only in the domain (denoted by $\mathcal{D}$) where they are smooth and strictly convex by the Freidlin-Wentzell uniqueness theorem of the orthogonal decomposition of the drift function $\vb$  \cite{freidlin1998random}.  Since the quasipotential $\varphi(\vx;\vx_f)$ is defined strictly by the trajectories of dynamics in Eq. \eqref{quasipotential}, to equate $ \hat{\varphi}(\vx)$ and $\varphi(\vx;\vx_f)$ on $\mathcal{D}$, we can justify the leading order term of the WKB ``ansatz" \eqref{wkb:stationary} for the invariant probability at least on $\mathcal{D}$. Analogously, $\vx_f$ can be extended from a fixed point to a invariant set, so the above statement is also true for the systems with a stable limit cycle \cite{freidlin1998random}. In the following work, we will restrict our analysis of dynamics contained in $\mathcal{D}$ and assume $\mathcal{D}$ is compact, and use one brief notation $\varphi(\vx)$ to represent both of the potentials. 

Based on the above justification of the time-invariant WKB ansatz \eqref{asymptotic2}, we can plug it into the stationary FPE \eqref{seq-PDE3} to get the system satisfied three equations \cite{Qian-2017}
\begin{eqnarray} 
& \mathbf{b}(\mathbf{x}) = -  \mathbf{D}\nabla \varphi(\mathbf{x}) + \vgamma(\mathbf{x}), \label{b.var.gamma.decomp}\\
    &\nabla\varphi(\mathbf{x}) \cdot \vgamma(\mathbf{x}) = 0, \label{b.var.gamma.decomp2}\\
    & \nabla \cdot (\omega (\mathbf{x}) \vgamma(\mathbf{x})) = - \nabla \varphi(\mathbf{x}) \cdot \mD\nabla \omega(\mathbf{x}). \label{divofomega}
\end{eqnarray}
Note that the vector field $ \mathbf{b}(\mathbf{x})$ represents deterministic dynamics and can be decomposed to two terms $ \nabla\varphi(\vx)  \perp \vgamma(\mathbf{x}) $ which is consistent with the FW's orthogonal decomposition \cite{freidlin1998random}. In comparison to the gradient flow $\nabla\varphi(\vx)$, dynamics following the vector field $\vgamma$ represents the part of circular motion of $\vb$. By the previous works \cite{holland1978stochastically,kurrer1991effect,vance1996fluctuations,ge2012landscapes,li2014landscape,lin2019quasi}, we have the following three propositions of $\varphi(\vx)$:
\begin{enumerate}
    \item  \label{prop-varphi-1} We have justified the leading order term $\varphi(\vx)$ in the stationary WKB ansatz. But it is based on the existence of the stationary probability distribution $\pi_\epsilon$. The existence of $\pi_\epsilon$ for stochastic stable limit cycles has been proved in the work of Holland \cite{holland1978stochastically}.
    \item  \label{prop-varphi-2} On a limit cycle, $\varphi(\vx)$ and $\nabla \varphi(\vx)$ are always zero. The landscape of  $\varphi(\vx)$ has a Mexican hat shape and the bottom of the Mexican hat ring characterizes the deterministic trajectory of the cycle \cite{li2014landscape}; And the second derivative of $\varphi(\vx)$ tangential to the cycle is also zero, which means the Gaussian fluctuations along the direction tangential to the cycle is eventually smeared out \cite{kurrer1991effect,vance1996fluctuations,ge2012landscapes}.  
    \item \label{prop-varphi-3} Along the cycle, the matrix $\nabla \nabla \varphi(\vx)$ is positive semi-definite \cite{vance1996fluctuations,lin2019quasi}. Specifically, the smallest eigenvalue of $\nabla \nabla \varphi(\vx)$ is zero on the cycle and the corresponding eigenvector is tangential to the cycle  (by the proposition \eqref{prop-varphi-2}). The rest of eigenvalues are positive, i.e., the Gaussian fluctuations perpendicular to the cycle  are outward and damped out by the dissipation toward the limit cycle. 
\end{enumerate}

The above is the setup of  $\varphi(\vx, t)$. Let us continue on analysis of the solution  of $\omega(\vx, t)$ in Eq. \eqref{eq:prefactor}. It can be rearranged as
\begin{align} \label{eq:cont-omega}
    \frac{\partial \omega(\vx,t)}{\partial t} + \nabla \cdot(\vb(\vx) \omega(\vx, t)) = -2 \mD \nabla\varphi(\vx, t) \nabla \omega(\vx, t) - \mD \nabla\nabla\varphi(\vx, t) \omega(\vx, t),
\end{align}
where $\mD \nabla\nabla\varphi(\vx, t)$ is the  the Frobenius product of the matrix $\mD$ and the matrix $\nabla\nabla\varphi(\vx, t)$. We can identify
Eq. \eqref{eq:cont-omega} is a {\em continuity equation} that describes the transport following the vector field $\hat{\vx}'=\vb(\hat{\vx})$ with a density-dependent sink (source) term on the right hand side. Therefore, the solution of Eq. \eqref{eq:cont-omega} gives us a measure of a large number particle system ``without" noise by the Eulerian description of dynamics $\hat{\vx}'=\vb(\hat{\vx})$. The original effect of noise, $\mD$, appears in the sink (source) term in this continuity equation of $\omega$. On the other hand, if we follow the dynamics of $\hat{\vx}'=\vb(\hat{\vx})$, we have the dynamics of $\omega$ by the Lagrangian description
\begin{align} \label{eq:rate-of-omega}
\frac{\rd \omega(\hat{\vx}(t), t)}{\rd t} &= \frac{\partial \omega(\hat{\vx}(t) ,t )}{\partial t} + \nabla \omega(\hat{\vx}(t), t ) \frac{\rd \hat{\vx}(t)}{\rd t} \notag  \\
 &= \nabla \cdot \vb(\vx) \omega(\hat{\vx}(t), t)  - \mD \nabla\nabla\varphi(\hat{\vx}(t), t)\omega(\hat{\vx}(t), t) \\
  &= - \nabla \cdot \vgamma(\vx(t) ,t) \omega(\hat{\vx}(t), t) \notag. 
\end{align}
To compare Eq. \eqref{eq:cont-omega} with Eq. \eqref{divofomega}, we have that $ \omega(\vx)$ in the WKB ansatz \eqref{wkb:stationary} is one of invariant solutions of Eq. \eqref{eq:cont-omega}. However, in distinction to the HJE of $\varphi(\vx, t)$, we shall notice that the uniqueness and the smoothness of invariant solutions of the PDE \eqref{eq:cont-omega} are not discussed in the present work.

Based on the above setup, we have the following theorem about the curvature of invariant large deviation rate function $ \nabla\nabla\varphi(\vx )$ and the logarithm of the prefactor $\ln\omega(\vx)$ along the dynamics $\vx^*(t)$ on the limit cycle $\Gamma$.  For this theorem, we require three assumptions of regular conditions:
\begin{enumerate}
    \item \label{assume-1} The functions $\varphi(\vx, t)$ and $\ln\omega(\vx,t)$ are smooth enough with respect to $\vx$, and the derivatives uniformly converge on $\mathcal{D}$ as $t \rightarrow \infty$, in which the compact domain $\mathcal{D}$ is defined above.
    \item \label{assume-1.5} The drift function $\vb(\vx)$ and its Jacobian $\mathbf{A}(\vx)$ are continuous on $\mathcal{D}$. 
    \item \label{assume-2} For all $t \geq 0$, the covariance matrix $\mSigma(t)$ in Lemma \ref{thm:two-types} is nonsingular, i.e, the Gaussian fluctuations of each direction is nonzero. Therefore, the inverse of Eq. \eqref{sigma-WKB}  in Lemma \ref{thm:two-types} is well-defined: $\left[\mSigma(t)\right]^{-1} = \nabla\nabla \varphi(\hat{\vx}(t), t)$ for all  $t \geq 0$.
\end{enumerate}

\begin{theorem}  \label{thm:dynamics-mSigma-prefactor}
Let $ \vx^*(t)$ be a deterministic trajectory with $\vx^*(0) \in \Gamma$,  $\left[ \mSigma^*(t)\right]^{-1} :=  \nabla\nabla\varphi(\vx^*(t) ) $, and $\omega^*(t) := \omega(\vx^*(t))$. For all $t > 0$,
\begin{align}
    \frac{ \rd \left[ \mSigma^*(t)\right]^{-1}}{\rd t} &= -\left[ \mSigma^*\right]^{-1} \mA(\vx^*)  - \mA(\vx^*)^T  \left[ \mSigma^*\right]^{-1}  -  2 \left[ \mSigma^*\right]^{-1} \mD \left[ \mSigma^*\right]^{-1},  \label{dynamics-cov-limit-cycle}\\
    \frac{\rd \ln\omega^*(t)}{\rd t} &= -\nabla \cdot \vb(\vx^*) - \mD \left[ \mSigma^*\right]^{-1} = -\nabla \cdot \vgamma(\vx^*). \label{dynamics-prefactor-limit-cycle}
\end{align}
Furthermore, the smallest eigenvalues of the matrix $\left[ \mSigma^*(t)\right]^{-1}$ is zero with the  eigenvector tangential to $\Gamma$ and the other eigenvalues are positive with the eigenvectors in the hyperplane perpendicular to $\Gamma$. 
\end{theorem}

\begin{proof}

In this proof, the norm  $\| \cdot   \|$ represents the supremum norm. By our setup of the diffusion process \eqref{macro.scale5}, the macroscopic deterministic trajectory $\hat{\vx}(t)$, with $\hat{\vx}(0) \in \mathcal{D}$,  converges to the stable limit cycle $\Gamma$, so we have
\begin{align} \label{asym-limit-cycle0}
    \lim_{t \rightarrow \infty} \min_{\vx^* \in \Gamma} \|  \hat{\vx}(t) - \vx^*   \| = 0.
\end{align}
By  the assumption \eqref{assume-1} and $\lim_{t \rightarrow \infty} \varphi(\vx, t) = \varphi(\vx) $ in the setup, we have that
\begin{align} \label{asym-limit-cycle0.5}
    \lim_{t \rightarrow \infty} \nabla\nabla\varphi(\vx, t) =  \nabla\nabla\varphi(\vx) \quad \text{uniformly on} \ \mathcal{D}.
\end{align}
Therefore, for any $\epsilon > 0$, there exists $T(\epsilon)>0$ such that for every $t > T(\epsilon)$, there is a point $\vx^*(t) \in \Gamma$ with its corresponding initial point $\vx^*(0) \in \Gamma$ and
\begin{align} \label{ext-varphi-limit-cycle}
    \|  \nabla\nabla\varphi(\hat{\vx}(t), t)  - \nabla\nabla \varphi(\vx^*(t))  \| =    O(\epsilon), 
\end{align}
where we use triangular inequality  
\begin{align} \label{asym-limit-cycle}
    \|  \nabla\nabla\varphi(\hat{\vx}(t), t)  - \nabla\nabla \varphi(\vx^*(t))  \| < \|  \nabla\nabla\varphi(\hat{\vx}(t), t)  - \nabla\nabla\varphi(\hat{\vx}(t))  \| + \|  \nabla\nabla\varphi(\hat{\vx}(t))  - \nabla\nabla \varphi(\vx^*(t))  \| 
\end{align}
with  Eq. \eqref{asym-limit-cycle0.5} for the first term and  Eq.  \eqref{asym-limit-cycle0}  for the second term on the right side of the inequality. Apply  $\left[\mSigma^*(t)\right]^{-1} := \nabla\nabla\varphi(\vx^*(t))$ defined in the theorem, and  $\left[\mSigma(t)\right]^{-1} = \nabla\nabla \varphi(\hat{\vx}(t), t) $ given in Lemma \ref{thm:two-types} with the assumption \eqref{assume-2}, to Eq. \eqref{ext-varphi-limit-cycle}, we further have that
\begin{align} \label{def:fluctuation-on-cycle2}
     \left\| \left[\mSigma^*(t)\right]^{-1} -
    \left[\mSigma(t)\right]^{-1} \right\| = O(\epsilon).
\end{align}
Following the same approach for the result \eqref{def:fluctuation-on-cycle2} with the assumption \eqref{assume-1} that the derivatives of $\varphi$ uniformly converge on $\mathcal{D}$ and the assumption \eqref{assume-1.5}  that $\vb$ is continuous on $\mathcal{D}$, we can also show that 
\begin{align} \label{def:fluctuation-on-cycle3}
     \left\| \frac{ \rd \left[\mSigma^*(t)\right]^{-1} }{\rd t} -
   \frac{ \rd \left[\mSigma(t)\right]^{-1} }{\rd t} \right\| = O(\epsilon).
\end{align}
Furthermore, by the Lyapunov equation of $\mSigma(t)$ obtained in Lemma \ref{thm:gaussian.z}, we have that
\begin{align} \label{varphi-limit-cycle0}
    \frac{ \rd \left[\mSigma(t)\right]^{-1} }{\rd t} = \left[\mSigma(t)\right]^{-1}\frac{ \rd \left[\mSigma(t)\right] }{\rd t} \left[\mSigma(t)\right]^{-1}=  -\left[\mSigma\right]^{-1} \mA(\hat{\vx})  - \mA(\hat{\vx})^T  \left[\mSigma\right]^{-1} -  2 \left[\mSigma\right]^{-1} \mD \left[\mSigma\right]^{-1}.
\end{align}
By the assumption \eqref{assume-1.5}  that $\mA$ is continuous on $\mathcal{D}$, combined with Eq.  \eqref{def:fluctuation-on-cycle2}, Eq. \eqref{def:fluctuation-on-cycle3} and Eq. \eqref{varphi-limit-cycle0}, 
we can show that
\begin{align} \label{varphi-limit-cycle}
    \frac{ \rd \left[\mSigma^*(t)\right]^{-1} }{\rd t} = -\left[\mSigma^*\right]^{-1} \mA(\vx^*)  - \mA(\vx^*)^T  \left[\mSigma^*\right]^{-1} -  2 \left[\mSigma^*\right]^{-1} \mD \left[\mSigma^*\right]^{-1} + O(\epsilon),
\end{align} 
where we use the the triangular inequality.
Eq. \eqref{varphi-limit-cycle} is true  for any $\epsilon > 0$ and $t > T(\epsilon)$, and furthermore,  the functions $\left[\mSigma^*(t)\right]^{-1} $ and $\mA(\vx^*(t)) $ are periodic by the definitions, so if Eq. \eqref{varphi-limit-cycle} holds for $t > T(\epsilon)$, by the phase shift, it should hold for all $t >0$. Then
we can take $\epsilon \rightarrow 0$ to obtain
\begin{align} \label{varphi-limit-cycle1}
    \frac{ \rd \left[\mSigma^*(t)\right]^{-1} }{\rd t} = -\left[\mSigma^*\right]^{-1} \mA(\vx^*)  - \mA(\vx^*)^T  \left[\mSigma^*\right]^{-1} -  2 \left[\mSigma^*\right]^{-1} \mD \left[\mSigma^*\right]^{-1}, 
\end{align}
for all $t>0$ with the initial condition $\left[\mSigma^*(0)\right]^{-1} = \nabla\nabla\varphi(\vx^*(0))$.

Next, we need to investigate the solution $\left[\mSigma^*(t)\right]^{-1}$ given by the ODE \eqref{varphi-limit-cycle1}. Let us introduce a coordinate transformation from the Cartesian coordinate to the curvilinear coordinate around the limit cycle by the change of basis
\begin{align} \label{coordinate-transformation}
   \mathbf{K}(t) = \mathbf{Q}(t)^{-1} [ \mSigma^*(t)]^{-1} \mathbf{Q}(t) \quad \text{with} \quad  \mathbf{Q}(t) = \left[ \mathbf{e}_1(t) \ \mathbf{e}_2(t) \ \cdots \ \mathbf{e}_n(t) \right],
\end{align}
in which $\mathbf{Q}(t)$ is a time-dependent orthonormal basis. In particular, $\mathbf{e}_1(t) = \vb(\vx^*)/ \| \vb(\vx^* \|$ is the tangential unit vector on $\Gamma$ and the span of the rest of vectors $\{ \mathbf{e}_2(t) \ \cdots \ \mathbf{e}_n(t)  \}$ represents the hyperplane perpendicular to $\Gamma$. For every fixed time $t$, $\mathbf{Q}(t)$ can be obtain by the {\em Gram–Schmidt process} and $\mathbf{Q}(t)$ is known as the {\em Frenet frame}. By the proposition \eqref{prop-varphi-2}, since $\nabla\nabla \varphi(\vx^*)$ is always zero on the direction tangential to $\Gamma$, $\mathbf{e}_1(t)$ is in the nullspace of $\mathbf{Q}(t)$ for all $t$. With the fact $ \left[\mSigma^*(t)\right]^{-1} := \nabla\nabla\varphi(\vx^*(t))$ is symmetric, for all $1\leq i \leq n$ and $1\leq j \leq n$, we have that 
\begin{align}
      \left[\mathbf{K}(t) \right]_{i1} =  \mathbf{e}_i(t)^T [ \mSigma^*(t)]^{-1} \mathbf{e}_1(t) \equiv 0  \quad \text{and } \quad   \left[\mathbf{K}(t) \right]_{1j} = \mathbf{e}_1(t)^T [ \mSigma^*(t)]^{-1}\mathbf{e}_j(t) \equiv 0,
\end{align}
thus the matrix $\mathbf{K}(t)$ has both zero first column and zero first row. Therefore, we can define a  submatrix of $\tilde{\mathbf{K}}(t)$ by deleting the first column and the first row of ${\mathbf{K}}(t)$, equipped with Eq. \eqref{varphi-limit-cycle}, we will obtain a $n-1$ by $n-1$ system of differential equations 
\begin{align} \label{varphi-limit-cycle2}
    \frac{ \rd \tilde{\mathbf{K}}(t)}{\rd t} &= -\tilde{\mathbf{K}}(t) \left[ \tilde{\mA}(\vx^*) - \tilde{\mathbf{S}}(t)\right] - \left[\tilde{\mA}(\vx^*) - \tilde{\mathbf{S}}(t) \right]^T  \tilde{\mathbf{K}}(t) -  2 \tilde{\mathbf{K}}(t)\tilde{\mD} \tilde{\mathbf{K}}(t),
\end{align}
in which the symbol $\sim $ on top of each matrix represents the  restriction of the original matrix  in the hyperplane perpendicular to $\Gamma$ and the additional term $\tilde{\mathbf{S}}(t)$, ${\mathbf{S}} = \mathbf{Q}^{-1}(t) \dot{\mathbf{Q}}(t) $, is from the coordinate transformation. We can identify that Eq. \eqref{varphi-limit-cycle2} is a {\em periodic Riccati differential equation}. Under the assumptions that $\Gamma$ is a stable limit cycle and $\tilde{\mD}$ is definite positive (the later follows from the definite positive $\mD$ in the setup \eqref{macro.scale5}), the solution of the periodic Riccati differential equation \eqref{varphi-limit-cycle2} has to be positive definite and periodic with the same period of the limit cycle. Mathematical analysis of this type of periodic Riccati differential equations  can be found in the  works \cite{bittanti1984stability,pastor1993differential,chen2000existence,zhou2011periodic} and a comprehensive numerical analysis with several examples was provided in the work \cite{lin2019quasi}. 

The proof for Eq. \eqref{dynamics-prefactor-limit-cycle} of $\ln\omega^*(t)$ follows the proof for Eq. \eqref{dynamics-cov-limit-cycle} of $\left[\mSigma^*(t)\right]^{-1} $. Apply the same asymptotic analysis (Eq. \eqref{asym-limit-cycle0} - Eq. \eqref{varphi-limit-cycle1}) to the dynamics of $\omega(\hat{\vx}(t), t)$ in Eq. \eqref{eq:rate-of-omega}, we can show that  $\ln\omega^*$ satisfies 
\begin{align} \label{omega-limit-cycle}
    \frac{\rd \ln\omega^*(t)}{\rd t} = \nabla \cdot \vb(\vx^*) - \mD \nabla\nabla\varphi(\vx^*) = -\nabla \cdot \vgamma(\vx^*),
\end{align}
for all $t > 0$ with the initial condition $ \ln\omega^*(0) = \ln \omega(\vx^*(0))$. 

\end{proof}

To the best of our knowledge, the asymptotic analysis from Eq. \eqref{asym-limit-cycle0} to Eq. \eqref{varphi-limit-cycle1} is the first work rigorously shows the Lyapunov differential equation \eqref{varphi-limit-cycle1} for the curvature of large deviation rate function near a deterministic trajectory from a transient state to an invariant set. Additionally, by analyzing the solution given by the Lyapunov differential equation  with the coordinate transformation, we confirm that the solution $\left[\mSigma^*(t)\right]^{-1}$ is consistent with the features of $\nabla \nabla \varphi(\vx)$ on $\Gamma$ in the proposition \eqref{prop-varphi-3} from the previous work. We will apply Theorem \ref{thm:dynamics-mSigma-prefactor} to further obtain three characterizations of stochastic limit cycles: (i) probability flux near the cycle (Sec. \ref{ch4-2}); (ii) two special features of the circular motion (Sec. \ref{ch4-3}); (iii) a local entropy balance equation on the cycle (Sec. \ref{ch4-4}).

\subsection{Probability flux of near a limit cycle} \label{ch4-2}

Recall that the vector field $\boldsymbol \gamma$ is defined by the orthogonal decomposition of the drift $\vb$ in Eq. \eqref{b.var.gamma.decomp} and Eq. \eqref{b.var.gamma.decomp2}, which characterizes the direction of circular motion of the deterministic dynamics $\hat{\vx}'=\vb(\hat{\vx})$. In addition to  the orthogonal decomposition of the drift $\vb$, $\vgamma$ can be derived from the probability flux $\mathbf{J}[p_\epsilon(\vx, t)]$ defined in the FPE \eqref{seq-PDE3}: 
\begin{align}
    \vgamma_\epsilon(\vx, t) := \frac{\mathbf{J}[p_\epsilon(\vx, t)]}{p_\epsilon(\vx,t)} = \vb(\vx) - \epsilon \mD \nabla \ln p_\epsilon(\vx, t), 
\end{align}
and take  $\epsilon \rightarrow 0$ before  $t \rightarrow \infty$, we have that 
\begin{align}
   \lim_{t \rightarrow \infty} \lim_{\epsilon \rightarrow 0} \vgamma_\epsilon(\vx, t) = \lim_{t \rightarrow \infty} \vgamma(\vx, t) = \vgamma(\vx),
\end{align}
in which $\vgamma(\vx, t)$ is the same one defined in the HJE \eqref{eq:rate-function}. For the stationary probability flux $\mathbf{J}[\pi_\epsilon(\vx)]$,
\begin{align} \label{def:gamma2}
   \vgamma_\epsilon(\vx) := \frac{\mathbf{J}[\pi_\epsilon(\vx)]}{\pi_\epsilon(\vx)} = \vb(\vx) - \epsilon\mD\nabla\ln\pi_\epsilon(\vx), 
\end{align}
which is corresponding to the reverse order of limits
\begin{align}
   \lim_{\epsilon \rightarrow 0} \lim_{t \rightarrow \infty} \vgamma_\epsilon(\vx, t) = \lim_{\epsilon \rightarrow 0}   \vgamma_\epsilon(\vx) = \vgamma(\vx).
\end{align}
Note that $\vgamma_\epsilon(\vx)$ has been recognized as Onsager's thermodynamics force \cite{qian2001mesoscopic}. Here we don't need to worry about the orders of limits if we just focus on the domain $\mathcal{D}$ based on our discussion in Sec. \ref{ch3-1}.

Recall that $\mA(\vx)$ is the Jacobian matrix of $\vb(\vx)$ and $\left[ \mSigma^*\right]^{-1} = \nabla\nabla\varphi(\vx^*(t))$. By the above setup, we have the following theorem for $\vgamma_\epsilon(\vx)$ and the stationary probability flux $\mathbf{J}[\pi_\epsilon(\vx)]$ near $\Gamma$.
\begin{theorem} \label{thm:prob-flux-limit-cycle}
For $ \vx^* \in \Gamma$ and $\|\vx - \vx^* \| = O(\sqrt{\epsilon})$, 
\begin{align} \label{prob-flux-limit-cycle}
    \vgamma_\epsilon(\vx) := \frac{\mathbf{J}[\pi_\epsilon(\vx)]}{\pi_\epsilon(\vx)} = \left[ \vb(\vx^*) + \left(\mA(\vx^*) + \mD\left[ \mSigma^*\right]^{-1}\right)(\vx - \vx^*) \right] + O(\epsilon),
\end{align}
where $ \mD\left[ \mSigma^*\right]^{-1}$ is the Frobenius product, and $\left[ \mSigma^*\right]^{-1}$ satisfies the equation 
\begin{align} \label{dynamics-cov-limit-cycle2}
    \frac{ \rd \left[ \mSigma^*(t)\right]^{-1}}{\rd t} = -\left[ \mSigma^*\right]^{-1} \mA(\vx^*)  - \mA(\vx^*)^T  \left[ \mSigma^*\right]^{-1}  -  2 \left[ \mSigma^*\right]^{-1} \mD \left[ \mSigma^*\right]^{-1}.
\end{align}

\end{theorem}

\begin{proof}
We first approximate $\nabla \ln \pi_\epsilon(\vx)$ near $\Gamma$, in which $\pi_\epsilon$ has the WKB expansion in Eq. \eqref{wkb:stationary}. Let $\|\vx - \vx^* \| = O(\sqrt{\epsilon})$, 
\begin{align} \label{approx-stationary-density}
    \nabla \ln \pi_\epsilon(\vx) &= \nabla \left( \ln \omega(\vx) - \frac{\varphi(\vx)}{\epsilon}\right) + O(\epsilon) \notag \\
    &= \nabla \left( \ln \omega(\vx) - \frac{(\vx - \vx^*)^T \nabla\nabla\varphi(\vx^*) (\vx - \vx^*)}{2\epsilon} \right) + O(1) \notag \\
    &=  - \frac{\nabla\nabla\varphi(\vx^*) (\vx - \vx^*)}{\epsilon} + O(1),
\end{align}
where we use $\varphi(\vx^*) \equiv 0$ and $\nabla \varphi(\vx^*) \equiv 0$. Apply the approximation \eqref{approx-stationary-density} to Eq. \eqref{def:gamma2}, we can further approximate $ \vgamma_\epsilon(\vx)$ around $\Gamma$
\begin{align} \label{prob.-velocity}
    \vgamma_\epsilon(\vx) := \frac{\mathbf{J}[\pi_\epsilon(\vx)]}{\pi_\epsilon(\vx)} &= \vb(\vx) - \epsilon\mD\nabla\ln\pi_\epsilon(\vx) \notag \\
    &= \vb(\vx^*) + \left(\mA(\vx^*) + \mD\left[ \mSigma^*\right]^{-1}\right)(\vx - \vx^*) + O(\epsilon).
\end{align}
The dynamics \eqref{dynamics-cov-limit-cycle2}  is from Eq. \eqref{dynamics-cov-limit-cycle} in Theorem \ref{thm:dynamics-mSigma-prefactor}.

\end{proof}

\begin{remark}
In the particular case of linear dynamics $\hat{\vx}'= \mA\hat{\vx}$ ($\mA$ is a constant $n \times n$ matrix) with a stable fixed point $\vx^*=0$, the stationary probability flux $\mathbf{J}$ has the formula \cite{qian2011nonlinear}
\begin{align} \label{prob-flux-fixed-point}
    \mathbf{J}[\pi_\epsilon(\vx)] = \pi^{-1}_\epsilon(\vx) \left[ \left(\mA +  \mD\mSigma^{-1} \right)\vx\right],
\end{align}
where $\mSigma^{-1}$ satisfies $\mSigma^{-1} \mA  +  \mA^T \mSigma^{-1} + 2\mSigma^{-1}\mD\mSigma^{-1} = 0 $. To compare Eq. \eqref{prob-flux-limit-cycle} with Eq. \eqref{prob-flux-fixed-point}, Theorem \ref{thm:prob-flux-limit-cycle} can be regarded as a local linear approximation of the stationary probability flux near each point $\vx^*$ on the limit cycle. This new result extends the case from a fix point to an invariant set. 
\end{remark}

\subsection{Two features of $\vgamma$ on the limit cycle} \label{ch4-3}

By Theorem \ref{thm:prob-flux-limit-cycle}, we have an approximate the probability flux near $\Gamma$, and furthermore, by the stationary FPE  \eqref{seq-PDE3}, there is another important property of the probability flux 
\begin{align} \label{div-free-flux}
     \nabla \cdot \mathbf{J}[ \pi_\epsilon(\mathbf{x}) ] = 0, \quad \text{for all} \ \mathbf{x} \in \mathbb{R}^n.
\end{align}
Since this property holds in the whole space, we can apply it to an arbitrary neighborhood of the limit cycle. Having this property, we will obtain two special features of $\vgamma$ on the limit cycle in this section.

The derivation of the system of equations \eqref{b.var.gamma.decomp} - \eqref{divofomega} in the work \cite{Qian-2017} was by plugging the WKB ansatz \eqref{wkb:stationary} into the stationary FPE and equating like-order terms of
\begin{align} \label{stationary.FPE}
    \nabla \cdot ( \pi_\epsilon(\mathbf{x}) \boldsymbol \gamma_\epsilon ( \mathbf{x}) ) = \nabla \cdot \mathbf{J}[ \pi_\epsilon(\mathbf{x}) ] = 0, \quad \text{for all} \ \mathbf{x} \in \mathbb{R}^n.
\end{align}
Applying Eq. \eqref{divofomega} in the system of equations to $\Gamma$, we obtain the first feature of $\vgamma$ on $\Gamma$
\begin{align} \label{div.free.gamma}
    \nabla \cdot (\omega(\mathbf{x}) \boldsymbol \gamma(\mathbf{x})  )  = - \nabla \varphi(\mathbf{x}) \cdot \mD\nabla \omega(\mathbf{x}) = 0  \quad \text{for all} \ \mathbf{x} \in \Gamma,
\end{align}
where we use $\nabla \varphi(\vx) \equiv 0$ for $\vx \in \Gamma$. We can recognize that the divergence-free stationary probability flux in $\mathbb{R}^n$ is the key to get Eq. \eqref{stationary.FPE} so that we can further obtain the divergence-free $\omega(\mathbf{x}) \boldsymbol \gamma(\mathbf{x}) $ on the limit cycle in Eq. \eqref{div.free.gamma}. With a similar approach, not only the divergence of $\omega(\mathbf{x}) \boldsymbol \gamma(\mathbf{x}) $, we can also obtain the second feature, $\lvert\lvert \omega(\mathbf{x}) \boldsymbol \gamma(\mathbf{x}) \rvert\rvert$, on the limit cycle via the following theorem: 
\begin{theorem} \label{thm:speed-omega} 
Let $1/v(\mathbf{x})$ be the product of the nonzero eigenvalues of the matrix $\nabla\nabla\varphi(\mathbf{x})$. Then
\begin{align}
    \sqrt{v(\mathbf{x})} \times \lvert\lvert \omega(\mathbf{x})   \vgamma(\mathbf{x}) \rvert\rvert
\end{align}
is constant on the limit cycle $\Gamma$. Furthermore, let $g_{\epsilon}(\mathbf{x})$ be the marginal density of $\pi_{\epsilon}(\vx)$ on the limit cycle $\Gamma$, then for $\mathbf{x} \in \Gamma$,
\begin{align} 
    g_\epsilon(\mathbf{x}) =  \frac{ \omega(\mathbf{x}) \sqrt{v(\mathbf{x})}}{ \int_{\Gamma}  \omega(\mathbf{y}) \sqrt{v(\mathbf{y})} d\mathbf{y}} + O(\epsilon),
\end{align}
and there exists a constant $K$ such that $g_\epsilon(\mathbf{x})\lvert\lvert \boldsymbol\gamma(\mathbf{x}) \rvert\rvert = K + O(\epsilon) $.
\end{theorem}

From the preceding discussion in Sec. \ref{ch3-1}, since $\varphi(\vx)$ is constant on $\Gamma$, the eigenvector of $\nabla\nabla\varphi(\mathbf{x})$ corresponding to the only one zero eigenvalue is tangential to $\Gamma$. Therefore, $v(\mathbf{x})$ in Theorem \ref{thm:speed-omega} defined on $\Gamma$ represents the scaled variance in the hyperplane perpendicular to $\Gamma$ (Recall that this hyperplane is defined by the span of the vectors $\mathbf{e}_2, \cdots, \mathbf{e}_n$ in the coordinate transformation \eqref{coordinate-transformation}). By Eq. \eqref{div.free.gamma}, we have that $\omega(\vx) \vgamma(\vx)$ is divergence-free on the limit cycle. By Theorem \ref{thm:speed-omega}, we further have that $\lvert\lvert\omega(\vx) \vgamma(\vx)\rvert\rvert$ is reciprocal to the scaled standard deviation perpendicular to the limit cycle. The later was mentioned in the previous work  \cite{vance1996fluctuations}. In the present work, we provide a mathematical proof in Appendix \ref{proof:constant-cycle}. The idea of proof is by using the Gauss's theorem for a tube around the limit cycle. Since the Gauss's theorem can only be applied to a small but finite tube, the divergence for the Gauss's theorem we use in the proof is $ \nabla \cdot ( \boldsymbol \gamma_\epsilon ( \mathbf{x}) \pi_\epsilon(\mathbf{x}))  = \nabla \cdot \mathbf{J}[ \pi_\epsilon(\mathbf{x}) ] = 0$, which holds for an arbitrary neighborhood of $\Gamma$.

\subsection{A local entropy balance equation on the limit cycle}

\label{ch4-4}

On the limit cycle, we have derived the local Gaussian fluctuations of dynamics  represented by the covariance $\mSigma^*(t)$ which follows the periodic Lyanupov equation  \eqref{dynamics-cov-limit-cycle} in Theorem \ref{thm:dynamics-mSigma-prefactor}. In general, the entropy of a Gaussian distribution $p$ with a covariance $ \mSigma$ is
\begin{align} \label{def:entropy}
    S = -\int p(\vx) \ln p(\vx) \rd \vx = \frac{1}{2} \ln\left[ 2\pi \text{det}(\mSigma) \right].
\end{align}
Therefore, in nonlinear stochastic systems, the ``local" entropy (denoted by $S_l$) due to the local Gaussian fluctuations has the rate defined by
\begin{align} \label{eq:entropy-rate}
    \frac{\rd S_l(t) }{ \rd t} := \frac{1}{2} \frac{\rd \ln \text{det}(\mSigma^*(t)) }{ \rd t }. 
\end{align}
By the property that the determinant of a matrix equals to the product of its eigenvalues, the local rate of entropy change \eqref{eq:entropy-rate} has an equivalent definition,
\begin{align} \label{eq:entropy-rate3}
    \frac{\rd S_l(t) }{ \rd t} := - \frac{1}{2} \frac{\rd\left( \sum_{k=1}^n\ln\lambda^*_k(t) \right)}{ \rd t }, 
\end{align}
where $\lambda^*_k(t), \ 1\leq k\leq n$ are the eigenvalues of $\left[\mSigma^*(t)\right]^{-1}$. Note that $1/v(\vx^*(t))$ defined in Theorem \ref{thm:speed-omega} equals to the product of all nonzero eigenvalues $\lambda^*_2(t) \cdots \lambda^*_n(t)$ of the matrix $\left[\mSigma^*(t)\right]^{-1}$.

By the above setup, we have the following theorem of a {\em local entropy balance equation} on the limit cycle $\Gamma$ with three equivalent expressions. 
\begin{theorem} \label{thm:entropy-change-limit-cycle}
For $ \vx^*(t) \in \Gamma$, by the definition \eqref{eq:entropy-rate} of the local rate of entropy change, there exists a local entropy balance equation with three equivalent expressions,
\begin{align}
    \frac{\rd S_l(t)}{ \rd t} &= \nabla \cdot \vgamma(\vx^*(t)), \label{rep-1}\\
    &= \displaystyle  -\frac{\rd\ln\omega(\mathbf{x^*}(t))}{\rd t},  \label{rep-2}\\
    &= \displaystyle  \frac{\rd \ln \lvert\lvert \boldsymbol\gamma(\vx^*(t)) \rvert\rvert }{ \rd t} + \frac{1}{2} \frac{\rd \ln v(\vx^*(t)) }{ \rd t}. \label{rep-3}
\end{align}
\end{theorem}

\begin{proof}
By Eq. \eqref{eq:entropy-rate}, with the dynamics of $\left[\mSigma^*(t)\right]^{-1}$ \eqref{dynamics-cov-limit-cycle}, we can obtain  
\begin{align} \label{eq:entropy-rate2}
     \frac{\rd S_l(t) }{ \rd t} = \nabla \cdot \vb(\vx^*) + \mD \left[\mSigma^*\right]^{-1} = \nabla \cdot \vgamma(\vx^*),
\end{align}
where $\mD \left[\mSigma^*\right]^{-1}$ is the  Frobenius product of the matrix $\mD$ and the matrix $\left[\mSigma^*\right]^{-1}$. Furthermore, by Eq. \eqref{dynamics-prefactor-limit-cycle} for the prefactor $\omega$, we can link Eq. \eqref{eq:entropy-rate2} to the dynamics of $\omega$,
\begin{align}
    \frac{\rd S_l(t) }{ \rd t} = -\frac{\rd\ln\omega(\mathbf{x^*}(t))}{\rd t}.\label{limitcycle:conserved} 
\end{align}

So far, we have proved the first two expressions \eqref{rep-1} and \eqref{rep-2}. The following proof is for the third expression \eqref{rep-3}: By Theorem \ref{thm:dynamics-mSigma-prefactor}, we know the smallest eigenvalue $\lambda^*_1(t) \equiv 0$ with its eigenvector tangential to the limit cycle. Therefore, the first term on the right side of Eq. \eqref{eq:entropy-rate3} requires a further analysis since
\begin{align} \label{eq:unclear}
    \frac{ \rd \ln \lambda^*_1(t)}{\rd t} = \frac{1}{ \lambda^*_1(t)} \frac{ \rd \lambda^*_1(t)}{\rd t} = \infty \times 0. 
\end{align}
To find an explicit formula of \eqref{eq:unclear}, we can use 
\begin{align} \label{pre-dynamics3}
    \frac{ \rd \left( \vgamma^T(\vx^*(t)) \left[\mSigma^*(t)\right]^{-1} \vgamma(\vx^*(t)) \right)}{\rd t} = 0, \quad \text{for all} \ t>0, 
\end{align}
since $\left[\mSigma^*(t)\right]^{-1}\vgamma(\vx^*(t)) \equiv \mathbf{0}$ by Theorem \ref{thm:dynamics-mSigma-prefactor}. By rearranging Eq. \eqref{pre-dynamics3}, we find a formula of Eq. \eqref{eq:unclear},  
\begin{align} \label{eq:gamma-lambda1}
     \frac{ \rd \ln \lambda^*_1(t)}{\rd t}  =  - 2 \frac{\rd \ln \|\vgamma(\vx^*(t)) \|}{\rd t},
\end{align}
where we use that $\lambda^*_1(t)$ is the eigenvalue of $\left[\mSigma^*(t)\right]^{-1}$ with respect to the eigenvector $\vgamma(\vx^*(t))$. Therefore,  by Eq. \eqref{eq:entropy-rate3} and Eq. \eqref{eq:gamma-lambda1}, and with the definition of $1/v(\vx^*(t)) := \prod_{k=2}^n \lambda_k^*(t)$, the local rate of entropy change on $\Gamma$ has another expression
\begin{align} \label{eq:entropy-rate4}
    \frac{\rd S_l(t) }{ \rd t} = \frac{\rd \ln \lvert\lvert \boldsymbol\gamma(\vx^*(t)) \rvert\rvert }{ \rd t} + \frac{1}{2} \frac{\rd \ln v(\vx^*(t)) }{ \rd t}.
\end{align}

\end{proof}

Each expression has a clear physical meaning:
 \begin{enumerate}
    \item The first expression \eqref{rep-1}: 
    The divergence of a vector field characterizes the volume change of the flow following this vector field. Therefore, the local entropy change can be considered as a consequence of volume-expanding (entropy-increasing) or volume-contracting (entropy-decreasing) of the circular flow $\vx^*(t)' = \vgamma(\vx^*)$. This expression of the entropy balance is corresponding to the {\em microscopic entropy production rate} given by a large number particle system  without noise \cite{dobbertin1976functional,steeb1979generalized}. 
    \item The second expression \eqref{rep-2}: Let us compare the rate of free energy change \cite{ge2012landscapes}  (the free energy is defined by the relative entropy $F(t) = \int_{\mathbb{R}^n} p(\vx,t) \log \left(\frac{p(\vx,t)}{\pi(\vx) }\right) \rd \vx$) with the local rate of entropy change on the limit cycle $\Gamma$:
\begin{align}
    \frac{\rd F(t)}{ \rd t} &= \frac{\rd \varphi({{\mathbf{x}^*}(t)})}{\rd t} \equiv 0, \\
     \frac{\rd S_l(t)}{ \rd t} &= -\frac{\rd\ln\omega(\mathbf{x}^*(t))}{\rd t} \label{rate-of-entropy-omega}.
\end{align}
The former follows the change of the large-deviation rate function $\varphi(\vx)$ on the deterministic trajectory, which is always zero on the limit cycle due to constant $\varphi(\vx^*(t))$; The later follows the change of $-\ln \omega(\vx^*(t))$, where the prefactor $\omega(\vx)$ is known as ``degeneracy" in the classical statistical mechanical terminology \cite{Qian-2017}, which is not constant on the limit cycle in general. 
\item The third expression \eqref{rep-3}: The local entropy balance equation can be decomposed into two parts
\begin{align} \label{eq:fluctuation-dissipation}
    \frac{\rd S_l(t)}{\rd t} = \underbrace{\frac{\rd \ln \lvert\lvert \boldsymbol\gamma(\vx^*(t)) \rvert\rvert }{ \rd t}}_{\text{dissipative part}} + \underbrace{\frac{1}{2} \frac{\rd \ln v(\vx^*(t)) }{ \rd t}}_{\text{fluctuation part}}.
\end{align}
The first part is yielded by the change of speed on the limit cycle, which is determined from the deterministic path of the dissipative dynamics $\vx^*(t)' = \vgamma(\vx^*)$;  The second part is constituted by the change of the Gaussian fluctuations perpendicular to $\Gamma$. The {\em fluctuation-dissipation theory of nonequilibrium systems} by Keizer \cite{keizer2012statistical} elucidated the relation between the fluctuations of a time-inhomogeneous Gaussian process and the associated dissipative deterministic path. Following this theory, Eq. \eqref{eq:fluctuation-dissipation} can be regarded as a {\em fluctuation-dissipation decomposition} of the local entropy balance equation on the limit cycle.

\end{enumerate}

\begin{remark}
By integrating the second expression \eqref{rep-2}, when the system reaches its steady state, we have an equation of local entropy near the limit cycle,
\begin{align} \label{formula-entropy}
    S_l(t) = - \ln\omega(\mathbf{x}^*(t)) + C,
\end{align}
for some constant $C$. By the equation of entropy \eqref{formula-entropy}, we know that in the long run, the  entropy of system measured near the limit cycle in the scope of the CLT should be {\em periodic} with the same period of the cycle. On the other hand, the global entropy in the total system has to be constant in the long run because of the existence of stationary distribution. 
\end{remark}

\begin{remark}
By the equivalence of the expressions \eqref{rep-2} and \eqref{rep-3} in Theorem \ref{thm:entropy-change-limit-cycle}, we have an alternative proof for the constant
$
    \sqrt{v(\mathbf{x})} \times \lvert\lvert \omega(\mathbf{x})   \vgamma(\mathbf{x}) \rvert\rvert
$
 on the limit cycle $\Gamma$ in Theorem \ref{thm:speed-omega}. 
\end{remark}

\section{Related issue: the scaling hypothesis of diffusion processes} 

\label{ch4}

As the success of the well-established scaling hypothesis in the continuous-time non-negative integer valued Markov population process $\mathbf{n}_V(t)$ (it has a law of large number as the system's size $V\rightarrow \infty: V^{-1}\mathbf{n}_V(t) \rightarrow \mathbf{c}(t)$, the concentration of all the species \cite{kurtz1972relationship}), we shall justify the origin of $\epsilon$ in \eqref{macro.scale5} with physical interpretations more than just a mathematical tool.

Let us begin with a diffusion process $\mathbf{Y}(\tau) \in \mathbb{R}^n$ satisfied the following SDE
\begin{align} \label{meso.scale}
    \rd \mathbf{Y}(\tau) = \mathbf{g}(\mathbf{Y})\rd \tau + [2\mathbf{D}(\mathbf{Y})]^{\frac{1}{2}} \rd \mB(\tau),
\end{align}
where $\vg: \mathbb{R}^n \rightarrow \mathbb{R}^n$ stands for the drift of the process, $\mD: \mathbb{R}^n \rightarrow \mathbb{R}^n \times \mathbb{R}^n$ is the diffusion matrix, and $\mB(\tau)$ is the standard $n$-dimensional Brownian motion. Through choosing different scales, $\vX=\mathbf{Y}/\alpha$, $t=\tau/\beta$, the SDE \eqref{meso.scale} can be rescaled as
\begin{align} \label{macro.scale}
    \rd \mathbf{X}(t) = \frac{\beta}{\alpha}\mathbf{g}(\alpha \vX) \rd t + \frac{\sqrt{\beta}}{\alpha} [2\mathbf{D}(\alpha\vX)]^{\frac{1}{2}} \rd \mB(t).
\end{align}
We assume a {\em space-time structure} $\beta = \xi(\alpha)$ by a function $\xi: \mathbb{R} \rightarrow \mathbb{R}$, and define a small parameter $\epsilon$
\begin{align} \label{def:epsilon}
    \epsilon := \xi(\alpha) / \alpha^2
\end{align}
with an implicit solution $\alpha^*(\epsilon)$ of Eq. \eqref{def:epsilon}. Under this framework, the scaled SDE \eqref{macro.scale} becomes a sequence of SDEs parameterized by $\epsilon$
\begin{align} \label{macro.scale2}
    \rd \mathbf{X}(t) = {\vb}_\epsilon(\vX) \rd t + [2\epsilon\mathbf{D}_\epsilon(\vX)]^{\frac{1}{2}} \rd \mB(t),
\end{align}
where
\begin{align}
    \vb_\epsilon(\vx) := \frac{\xi(\alpha^*(\epsilon))}{\alpha^*(\epsilon)} \vg\left(\alpha^*(\epsilon)\vx\right) \quad \text{and} \quad \mD_\epsilon(\vx) := \mD\left(\alpha^*(\epsilon)\vx\right).
\end{align}
In order to observe an emergent phenomenon as $\epsilon \rightarrow 0$, we require certain conditions (i) $\lim_{\epsilon\rightarrow 0}{\vb}_\epsilon(\vx)=\vb(\vx)$ exists. (ii) $\lim_{\epsilon\rightarrow 0}\mD_\epsilon(\vx) = \mD(\vx)$  exists. (iii) The convergence of random processes solved of the SDEs \eqref{macro.scale2} in different modes exists \cite{freidlin1998random}. Then the limit gives us an emergent deterministic dynamics $\rd \vx(t)/ \rd t = \vb(\vx)$ as $\epsilon \rightarrow 0$. 

In connection to classical overdamped mechanical motions in a viscous fluid, Eq. \eqref{macro.scale2} is widely called a Langevin equation, and in this case the small parameter $\epsilon$ has been identified as related to temperature of the system as well as the “scale” under which the mechanical motion is being observed \cite{van1992stochastic}. In reality, the limit of $\epsilon$ being zero should be interpreted as particle motions in a ``continuous medium at finite temperature" rather than ``temperature asymptotic to zero". We follow this physical intuition and so called {\em scaling hypothesis} \cite{kardar2007statistical} for the origin of $\epsilon$. The following discussion offers an insight into the connection between our scaling hypothesis for diffusion processes and the scaling hypothesis for statistical physics of fields.

The space-time structure defined by the function $\xi$ is rather general. To illustrate our hypothesis, we focus on a specific space-time structure $\beta = \alpha^k$ and thus the small parameter is defined as $\epsilon := \alpha^{k-2}$. In this example, when $k<2$, deterministic dynamics emerges at the macroscopic scale ($\alpha \rightarrow \infty$); when $k>2$, deterministic dynamics emerges at the microscopic scale ($\alpha \rightarrow 0$). The choice of $k$ depends on the property of the underlying drift function $\vg$ and the diffusion function $\mD$. Given $\vg(\vx) = c \vx^n$, $c$ is a constant, and $\mD$ is a constant matrix, then the sequence of SDEs \eqref{macro.scale2} becomes
\begin{align} \label{eqn-homogeneity}
    \rd \vX(t) = \epsilon^{\frac{k-1+n}{k-1}} c \vX^n \rd t + [2\epsilon\mathbf{D}]^{\frac{1}{2}} \rd \mB(t).
\end{align}
In order to fulfill the conditions of convergence, in this example, the order of space-time structure $k$ must be determined by the order of the underlying drift function $n$, i.e., $k = 1-n$. Hence, the scaled drift function becomes $\epsilon$-independent while the diffusion term is asymptotic to zero as $\epsilon \rightarrow 0$, which gives rise to an emergent deterministic dynamics $\rd \vx(t) = c \vx^n \rd t$. In other words, as a scientific theory, when we are able to observe deterministic dynamics $\rd \vx(t) = c \vx^n \rd t$ in a ``macroscopic" experiment, with an underlying stochastic dynamics having the drift function $\vg(\vx) = c \vx^n$, $n > -1$, this experiment must be running by the right space-time structure $\beta = \alpha^{1-n}$.

In addition to the scale of space, by the space-time relation $\beta = \alpha^{1-n}$, for the macroscopic emergent deterministic dynamics ($\alpha \rightarrow \infty$), there is a corresponding scale of time for emergent laws: As $ -1 < n < 1$, the deterministic dynamics emerges in a long-time limit; As $ n > 1$, it emerges in a short-time limit. So emergent dynamics could be observed at different combinations of space-time scales, which are determined by $n$.  This scaling exponent is given by the drift function $\vg$ of the underlying diffusion process.  Therefore, in our scaling hypothesis for diffusion processes, experimental observation of a power law for the space-time structure is determined by the underlying physics. So we name it {\em scaling hypothesis}, which upholds the principle of scaling hypothesis for statistical physics of fields \cite{kardar2007statistical}. 

Here we want to introduce two types of celebrated theories which inspired our scaling hypothesis and point out what the new results we can provide beyond those theories:
\begin{enumerate}
    \item  As we mentioned in Sec. \ref{ch1-2}, the sequence of SDEs \eqref{macro.scale2} has been carefully studied in the text {\em random perturbations of dynamical systems} by Freidlin and Wentzell \cite{freidlin1998random}. Our scaling hypothesis gives $\epsilon$ a physical meaning which was unclear. 
    \item The Kurtz's first theorem \cite{kurtz1972relationship} showed that the ODE model is an emergent model under the infinite volume limit of the discrete Markov chain model; And the Kurtz's second theorem \cite{kurtz1981central} about the CLT for Markov chains is a generalization of a simple random walk for Donsker's invariance principle \cite{donsker1951invariance}. The main distinction between our scaling hypothesis and the scaling used in the Kurtz's theorems is that the former is for a sequence of {\em scaled stochastic differential equations} but the later is for a sequence of the sum over a {\em  scaled Markov chain}.
\end{enumerate}

\section{Discussions and applications}
\label{ch5}

In Sec. \ref{ch2}, we provided two preliminaries, Eq. \eqref{dyn-cov1} in Lemma   \ref{thm:gaussian.z} and Eq. \eqref{sigma-WKB} in Lemma \ref{thm:two-types}: The former gives us the dynamics of the covariance of a time-inhomogeneous Gaussian process and the later characterizes the local curvature of the time-dependent large deviation function near it infimum. They both have a nice property that the existence of stationary probability is not required, so it helps us understand transient behaviors of the systems whose stationary probability may not exist. For example, if a system has unstable macroscopic deterministic dynamics, we can still compute its transient local Gaussian fluctuations and curvature of the rate function near the deterministic trajectory.

In Sec. \ref{ch3}, by asymptotic analysis, we characterized the dynamics near a stable limit cycle, and we found that the prefactor $\omega$ in the WKB ansatz plays an important role, which can be seen in Theorem \ref{thm:speed-omega} and Theorem \ref{thm:entropy-change-limit-cycle}. In contradistinction to the well-established theories \cite{evans1985pde,feng2006large} of the HJE \eqref{eq:rate-function} for the large deviation rate function $\varphi$, to the best of our knowledge,  sophisticated mathematical analysis of the PDE \eqref{eq:rate-of-omega} for $\omega$ might be missing and it is worthy of attention in the future. For applications, the local entropy balance equation in Theorem \ref{thm:entropy-change-limit-cycle} can help us seek a better understanding of thermodynamic behaviors of stochastic biological oscillators, e.g.,  (i)  mammalian cell cycles under external noises \cite{li2014landscape}, (ii) a modified Morris–Lecar conductance-based model of a neuron driven by extrinsic noise \cite{bressloff2018variational}, and (iii) Rosenzweig-MacArthur model for predator-prey interactions with the effect of stochasticity \cite{mendler2018analysis}.

In Sec. \ref{ch4}, the scaling hypothesis as a scientific theory, it allows us to apply the treatment to mathematical models of the complex systems whose “noise” does not have a clear origin as classical overdamped mechanical motions in a viscous fluid described by the Langevin equation. For example, the hypothesis could be used to the models of mechanical motions in biology with noise due to coarse graining. In addition to the justification of small parameter $\epsilon$ itself, this hypothesis may give us a clarification of the origin of $\epsilon$-dependent drift function $\vb$. It is known that there are two types of integrals for SDEs:
\begin{align}
    \rd \vX(t) &= \vb_I(\vX)\rd t + [2\epsilon\mathbf{D}(\vX)]^{\frac{1}{2}} \rd \mB(t) \quad \text{(It\^o interpretation)},\\
    \rd \vX(t) &= \vb_S(\vX)\rd t + [2\epsilon\mathbf{D}(\vX)]^{\frac{1}{2}} \circ \rd \mB(t) \quad \text{(Stratonovich interpretation)}.
\end{align}
The former is commonly used in mathematical analysis and financial mathematics and the later is mostly applied in physics and engineering. Note that 
\begin{align} \label{I-S}
    \vb_I(\vx) = \vb_S(\vx) + \epsilon \nabla \cdot \mathbf{D}(\vx). 
\end{align}
Follow the scaling hypothesis, the existence of the extra $\epsilon$-order term in Eq. \eqref{I-S}  could be a corollary of the existence of higher-order terms in the underlying drift function before scaling. This hypothesis provides us a  link between the two types of integral. It might help us to unravel the mystery of It\^o - Stratonovich dilemma \cite{van1992stochastic,gardiner2009stochastic,ao2004potential} in the future.

\section{Appendix}

\subsection{Proofs of Lemma \ref{lemma-2} and Lemma \ref{lemma-3} }
\label{proof:lemmas}

\subsubsection{Proof of Lemma \ref{lemma-2} }
\begin{proof}
For the special case of $\vx\in\mathbb{R}$, $\Theta=3$, $\Lambda=15$, and $\mXi=[h''(\vx^*)]^{-1}$.  
Then Eq. (\ref{eq-13-c}) 
\[
   \eta(\vx^*) =  \frac{f''(\vx^*)}{2h''(\vx)}-\left[
     \frac{f'(\vx^*)h'''(\vx^*)}{2[h''(\vx^*)]^2}+
            \frac{f(\vx^*)h''''(\vx^*)}{8[h''(\vx^*)]^2}\right]
   + \frac{5f(\vx^*)[h'''(\vx^*)]^2}{24[h'''(\vx^*)]^3}.
\]
This result can be found on P. 273 of \cite{bender2013advanced}, 
Eq. (6.4.35).   

	For the general case, 
\begin{eqnarray}
	 && \int_{\mathbb{R}^n} f(\vx)e^{-\frac{h(\vx)}{\epsilon}}\rd\vx
\nonumber\\
   &=& \int_{\mathbb{R}^n} 
       \left[ f(\vx^*) + (\vx-\vx^*)\cdot\nabla f(\vx^*) 
         +\frac{(\vx-\vx^*)^T\nabla\nabla f(\vx^*)(\vx-\vx^*)}{2}
          + \dots \right] 
\nonumber\\
	&\times&  \exp\left[-\frac{h(\vx^*)}{\epsilon}-\frac{(\vx-\vx^*)_ih''_{ij}(\vx^*) (\vx-\vx^*)_j}{2\epsilon} -\frac{h'''_{ijk}(\vx^*)(\vx-\vx^*)_i 
               (\vx-\vx^*)_j(\vx-\vx^*)_k}{6\epsilon} \right.
\nonumber\\
	&-&  \left. \frac{h''''_{ijk\ell}(\vx^*)(\vx-\vx^*)_i
       (\vx-\vx^*)_j(\vx-\vx^*)_k(\vx-\vx^*)_{\ell}}{24\epsilon} + \cdots  
       \right]  \rd\vx
\nonumber\\
	&=&  e^{-\frac{h(\vx^*)}{\epsilon}}
     \int_{\mathbb{R}^N} 
       \left[ f(\vx^*) +(\vx-\vx^*)\cdot\nabla f(\vx^*) + 
             \frac{(\vx-\vx^*)^T\nabla\nabla f(\vx^*) (\vx-\vx^*)}{2}
          + \dots \right] 
\nonumber\\
	&\times&  \left[ 1 -\frac{h'''_{ijk}(\vx^*)}{6\epsilon}(\vx-\vx^*)_i 
                   (\vx-\vx^*)_j(\vx-\vx^*)_k
             - \frac{h''''_{ijk\ell}(\vx^*)}{24\epsilon}(\vx-\vx^*)_i 
               (\vx-\vx^*)_j \right.
\nonumber\\ 
  &\times& \left. (\vx-\vx^*)_k(\vx-\vx^*)_{\ell}
          + \frac{\big[h'''_{ijk}(\vx^*)(\vx-\vx^*)_i(\vx-\vx^*)_j(\vx-\vx^*)_k\big]^2}{72\epsilon^2}
    +\cdots  
       \right]
\nonumber\\
	&\times&  e^{-\frac{ (\vx-\vx^*)^T\nabla\nabla h(\vx^*) (\vx-\vx^*)}{2\epsilon}
             } \rd\vx
\nonumber\\
	&=&  e^{-\frac{h(\vx^*)}{\epsilon}}
     \int_{\mathbb{R}^n}  e^{-\frac{(\vx-\vx^*)^T\nabla\nabla h(\vx^*)(\vx-\vx^*) }{2\epsilon}
            }
       \left[ f(\vx^*) + (\vx-\vx^*)\cdot\nabla f(\vx^*)  \right.
\nonumber\\
	&+& \frac{ (\vx-\vx^*)^T\nabla\nabla f(\vx^*) (\vx-\vx^*)}{2}
	-  \frac{ f(\vx^*)h'''_{ijk}(\vx^*)}{6\epsilon}(\vx-\vx^*)_i
       (\vx-\vx^*)_j(\vx-\vx^*)_k
\nonumber\\
    &-& \left( \frac{f'_i(\vx^*)h'''_{jk\ell}(\vx^*)}{6\epsilon}  + 
            \frac{f(\vx^*)h''''_{ijk\ell}(\vx^*)}{24\epsilon}
              \right) (\vx-\vx^*)_i  (\vx-\vx^*)_j(\vx-\vx^*)_k(\vx-\vx^*)_{\ell}  
\nonumber\\
	&+&    \left. 
            \frac{f(\vx^*)[h'''_{ijk}(\vx^*)(\vx-\vx^*)_i(\vx-\vx^*)_j(\vx-\vx^*)_k ]^2}{72\epsilon^2} + \cdots      \right] \rd\vx
\label{eq-100}\\
	&=&  \sqrt{\frac{(2\pi\epsilon)^N}{\det\big[\nabla\nabla h(x^*)\big]}}  e^{-\frac{h(\vx^*)}{\epsilon}}
       \left\{ f(\vx^*)  + \frac{\epsilon f''_{ij}(\vx^*)}{2}\Xi_{ij}(\vx^*)  \right.
\nonumber\\
   &-&
     \epsilon\left[ \frac{f'_i(\vx^*)h'''_{jk\ell}(\vx^*)}{6}  + 
            \frac{f(\vx^*)h''''_{ijk\ell}(x^*)}{24}
              \right] \Xi^{\frac{1}{2}}_{i\mu} 
       \Xi^{\frac{1}{2}}_{j\nu}\Xi^{\frac{1}{2}}_{k\rho}
            \Xi^{\frac{1}{2}}_{\ell\kappa}\Theta_{\mu\nu\rho\kappa}  
\nonumber\\
	&+&    \left. \epsilon 
      \left( \frac{f(\vx^*)[h'''_{ijk}(\vx^*)]^2}{72} \right)
         \Xi^{-\frac{1}{2}}_{i\mu}
         \Xi^{-\frac{1}{2}}_{i\mu'}\Xi^{-\frac{1}{2}}_{j\nu}\Xi^{-\frac{1}{2}}_{j\nu'}\Xi^{-\frac{1}{2}}_{k\rho}\Xi^{-\frac{1}{2}}_{k\rho'} 
        \Lambda_{\mu\mu'\nu\nu'\rho\rho'} +  \cdots \right\}.
\nonumber
\end{eqnarray}

	A multivariate normal distribution with covariance matrix
$\mXi$, which is positive definite thus 
$\mXi=\mXi^{\frac{1}{2}}\mXi^{\frac{T}{2}}$ 
\cite{bleistein1986asymptotic}, has
\begin{eqnarray*}
    && \frac{1}{\big[ \big(2\pi\epsilon\big)^N \det\big(\mXi\big)\big]^{\frac{1}{2}} } \int_{\mathbb{R}^n}
     f''_{ij}({\bf 0})x_ix_j
  \exp\left[-\frac{1}{2\epsilon}\vx^T \mXi^{-1} \vx \right]\rd\vx
\\
   &=& \frac{\epsilon\Xi^{\frac{1}{2}}_{i\nu} \Xi^{\frac{1}{2}}_{j\mu}    
 f''_{ij}({\bf 0})}{ (2\pi)^{N/2} }
        \int_{\mathbb{R}^n}y_{\nu}y_{\mu}
  \exp\left[-\frac{\vy^T\vy}{2}  \right]\rd\vy 
      = \epsilon f''_{ij}({\bf 0})\Xi_{ij},
\end{eqnarray*}
the Frobenius product of the Hessian matrix and 
covariant matrix $\mXi$,
\begin{eqnarray*}
  && \Big[ \big(2\pi\epsilon\big)^n \det\big(\mXi\big)\Big]^{-\frac{1}{2}}
	 \int_{\mathbb{R}^n} 
     f'''_{ijk}({\bf 0})x_ix_jx_k
\exp\left[-\frac{1}{2\epsilon}\vx^T \mXi \vx \right] \rd\vx
  =  0,
\\
  && \Big[ \big(2\pi\epsilon\big)^n \det\big(\mXi\big)\Big]^{-\frac{1}{2}}
	 \int_{\mathbb{R}^n} 
     f''''_{ijk\ell}({\bf 0})x_ix_jx_kx_{\ell}
\exp\left[-\frac{1}{2\epsilon}\vx^T \mXi \vx \right] \rd\vx
\nonumber\\
   &=& \epsilon^2 f''''_{ijk\ell}({\bf 0})
     \Xi^{\frac{1}{2}}_{i\mu}\Xi^{\frac{1}{2}}_{j\nu}
    \Xi^{\frac{1}{2}}_{k\rho}\Xi^{\frac{1}{2}}_{\ell\kappa}
       \Theta_{\mu\nu\rho\kappa},
\\
 && \Big[ \big(2\pi\epsilon\big)^n \det\big(\mSigma\big)\Big]^{-\frac{1}{2}}
	 \int_{\mathbb{R}^n} 
     f'''_{ijk}({\bf 0})x_i^2x_j^2x_k^2
\exp\left[-\frac{1}{2\epsilon}\vx^T \mXi \vx \right] \rd\vx
\nonumber\\
  &=&  \epsilon^3\big(2\pi\big)^{-\frac{n}{2}}
     f'''_{ijk}({\bf 0}) \Xi^{-\frac{1}{2}}_{i\mu}
         \Xi^{-\frac{1}{2}}_{i\mu'}\Xi^{-\frac{1}{2}}_{j\nu}\Xi^{-\frac{1}{2}}_{j\nu'}\Xi^{-\frac{1}{2}}_{k\rho}\Xi^{-\frac{1}{2}}_{k\rho'}
 \int_{\mathbb{R}^n} 
y_{\mu}y_{\mu'}y_{\nu}y_{\nu'}y_{\rho}
y_{\rho'}
\exp\left[-\frac{\vy^T\vy}{2}\ \right] \rd\vx
\nonumber\\
   &=& \epsilon^3 f'''_{ijk}({\bf 0})
   \Xi^{-\frac{1}{2}}_{i\mu}
         \Xi^{-\frac{1}{2}}_{i\mu'}\Xi^{-\frac{1}{2}}_{j\nu}\Xi^{-\frac{1}{2}}_{j\nu'}\Xi^{-\frac{1}{2}}_{k\rho}\Xi^{-\frac{1}{2}}_{k\rho'}
       \Lambda_{\mu\mu'\nu\nu'\rho\rho'}.
\end{eqnarray*}

	Applying (\ref{eq-13-a}) to
both numerator and denominator of the lhs of (\ref{eq-13-b}),
\[
	 \frac{ 
     f(\vx^*)+\epsilon\left\{ \begin{array}{c} \frac{f''_{ij}(\vx^*)\Xi_{ij}}{2}-\left[
     \frac{f'_i(\vx^*)h'''_{jk\ell}(\vx^*)}{6}+
            \frac{f_{\alpha}(\vx^*)h''''_{ijk\ell}(\vx^*)}{24}\right]
    \Xi^{\frac{1}{2}}_{i\mu} 
       \Xi^{\frac{1}{2}}_{j\nu}\Xi^{\frac{1}{2}}_{k\rho}
            \Xi^{\frac{1}{2}}_{\ell\kappa}\Theta_{\mu\nu\rho\kappa} \\
           + \frac{f(\vx^*)[h'''_{ijk}(\vx^*)]^2}{72} \Xi^{-\frac{1}{2}}_{i\mu}
         \Xi^{-\frac{1}{2}}_{i\mu'}\Xi^{-\frac{1}{2}}_{j\nu}\Xi^{-\frac{1}{2}}_{j\nu'}\Xi^{-\frac{1}{2}}_{k\rho}\Xi^{-\frac{1}{2}}_{k\rho'} 
        \Lambda_{\mu\mu'\nu\nu'\rho\rho'} \end{array} \right\} 
         + O(\epsilon^2)}{ 1+\epsilon\left\{ \begin{array}{c} -\left[
            \frac{h''''_{ijk\ell}(\vx^*)}{24}\right]
    \Xi^{\frac{1}{2}}_{i\mu} 
       \Xi^{\frac{1}{2}}_{j\nu}\Xi^{\frac{1}{2}}_{k\rho}
            \Xi^{\frac{1}{2}}_{\ell\kappa}\Theta_{\mu\nu\rho\kappa} \\
           + \frac{[h'''_{ijk}(\vx^*)]^2}{72} \Xi^{-\frac{1}{2}}_{i\mu}
         \Xi^{-\frac{1}{2}}_{i\mu'}\Xi^{-\frac{1}{2}}_{j\nu}\Xi^{-\frac{1}{2}}_{j\nu'}\Xi^{-\frac{1}{2}}_{k\rho}\Xi^{-\frac{1}{2}}_{k\rho'} 
        \Lambda_{\mu\mu'\nu\nu'\rho\rho'} \end{array} \right\}
            + O(\epsilon^2)}
\]
\[
	=  f(\vx^*) + \epsilon \left[\frac{f''_{ij}(\vx^*)\Xi_{ij}}{2}-
     \frac{f'_i(\vx^*)h'''_{jk\ell}(\vx^*)\Xi^{\frac{1}{2}}_{i\mu} 
       \Xi^{\frac{1}{2}}_{j\nu}\Xi^{\frac{1}{2}}_{k\rho}
           \Xi^{\frac{1}{2}}_{\ell\kappa}\Theta_{\mu\nu\rho\kappa}}{6}  \right]  +
         O(\epsilon^2).
\]

\end{proof}

\subsubsection{Proof of Lemma \ref{lemma-3} }

We only provide the proof for the case $\vx \in \mathbb{R}^1$, which is denoted by $x$. For higher dimensions, results are the same by using the notations from Lemma \ref{lemma-2}.
\begin{proof}
Let the global minimum of 
$[h(x)-\epsilon\ln g(x)]$ be at $\tilde{x}^*=x^*+\Delta x(\epsilon)$.
Clearly, $\Delta x\to 0$ as $\epsilon\to 0$.  In fact,
\begin{eqnarray*}
   &&  \left[  h'(x)-\epsilon\left(\frac{g'(x)}{g(x)}\right) \right]_{x=x^*+\Delta x}
      = 0,
\\
	&&  h'(x^*)+h''(x^*)\Delta x 
   -\epsilon\left(\frac{g'(x^*)}{g(x^*)}\right)
      = 0,
\\
	 && \Delta x = \epsilon\left(\frac{g'(x^*)}{g(x^*)h''(x^*)}\right)
        +O\big(\epsilon^2\big).
\end{eqnarray*}
Now we apply the Eq. (\ref{eq-13-b}) in Lemma \ref{lemma-2}:
\begin{eqnarray}
    && \frac{\displaystyle \int_{-\infty}^{\infty}
                     f(x)g(x)e^{-\frac{h(x)}{\epsilon}}\rd x}{
            \displaystyle  \int_{-\infty}^{\infty}
  g(x) e^{-\frac{h(x)}{\epsilon}}\rd x }  = \frac{\displaystyle \int_{-\infty}^{\infty}
                    f(x)e^{-\frac{h(x)-\epsilon\ln g(x)}{\epsilon}}\rd x}{
            \displaystyle  \int_{-\infty}^{\infty}
   e^{-\frac{h(x)-\epsilon\ln g(x)}{\epsilon}}\rd x }
\nonumber\\
    &=& f( \tilde{x}^*) + \epsilon \left[\frac{f''(\tilde{x}^*)}{2h''(\tilde{x}^*)}-
     \frac{f'(\tilde{x}^*)h'''(\tilde{x}^*)}{2[h''(\tilde{x}^*)]^2}  \right]  +
         O(\epsilon^2),
\nonumber\\
  &=& f(x^*) +\epsilon \underbrace{ \left(\frac{ f'(x^*)g'(x^*)}{g(x^*)h''(x^*)} \right. }_{ \text{ $f'(x^*)\Delta x$ due to } g(x) } +  \underbrace{ \frac{f''(x^*)}{2h''(x^*)} }_{ \text{ due to } f(x) } -
     \underbrace{ \left. \frac{f'(x^*)h'''(x^*)}{2[h''(x^*)]^2} \right) 
    }_{ \text{ due to non-quadratic } h(x) }  +
         O(\epsilon^2).
\end{eqnarray}
For the terms on the order of $\epsilon$, replacing
$\tilde{x}^*$ by $x^*$ only affects the order
$\epsilon^2$ term.

\end{proof}

\subsubsection{Proof of Lemma \ref{lemma-4} }

Here we provide an additional Lemma, which is not used in the present work, but it is useful for related fields.   
\begin{lemma}
\begin{equation}
  \frac{\displaystyle \int_{-\infty}^{\infty}
                     f^2(x)g(x)e^{-\frac{h(x)}{\epsilon}}\rd x}{
            \displaystyle  \int_{-\infty}^{\infty}
  g(x) e^{-\frac{h(x)}{\epsilon}}\rd x } - \left[
    \frac{\displaystyle \int_{-\infty}^{\infty}
                     f(x)g(x)e^{-\frac{h(x)}{\epsilon}}\rd x}{
            \displaystyle  \int_{-\infty}^{\infty}
  g(x) e^{-\frac{h(x)}{\epsilon}}\rd x } \right]^2
  =  \epsilon\left(  \frac{f'^2(x^*)}{h''(x^*)}  \right) + O\big(\epsilon^2\big).
\end{equation}
\label{lemma-4}
\end{lemma}

\begin{proof}
\begin{eqnarray}
  && \frac{\displaystyle \int_{-\infty}^{\infty}
                     f^2(x)g(x)e^{-\frac{h(x)}{\epsilon}}\rd x}{
            \displaystyle  \int_{-\infty}^{\infty}
  g(x) e^{-\frac{h(x)}{\epsilon}}\rd x } - \left[
    \frac{\displaystyle \int_{-\infty}^{\infty}
                     f(x)g(x)e^{-\frac{h(x)}{\epsilon}}\rd x}{
            \displaystyle  \int_{-\infty}^{\infty}
  g(x) e^{-\frac{h(x)}{\epsilon}}\rd x } \right]^2
\nonumber\\
  &=&  \epsilon 
     \left(\frac{ 2f(x^*)f'(x^*)g'(x^*)}{g(x^*)h''(x^*)}+ 
         \frac{2f'^2(x^*)+2f(x^*)f''(x^*) }{2h''(x^*)}-\frac{2f(x^*)f'(x^*)h'''(x^*)}{2[h''(x^*)]^2}  \right)
\nonumber\\
	&-& 2\epsilon f(x^*)
     \left(\frac{ f'(x^*)g'(x^*)}{g(x^*)h''(x^*)}+ 
         \frac{ f''(x^*) }{2h''(x^*)}-\frac{f'(x^*)h'''(x^*)}{2[h''(x^*)]^2}  \right)
         + O\big(\epsilon^2\big)
\nonumber\\
	&=& \epsilon\left(  \frac{f'^2(x^*)}{h''(x^*)}  \right) + O\big(\epsilon^2\big).
\end{eqnarray}

\end{proof}

\subsection{Proof of Theorem \ref{thm:speed-omega} } \label{proof:constant-cycle}

\begin{proof}
Given $\mathbf{x_1} \in \Gamma$, let $\mathcal{P}_1$ be a plane containing $\mathbf{x_1}$ and perpendicular to the vector $\boldsymbol\vgamma(\mathbf{x_1})$. Given another point $\mathbf{x}_2 \in \Gamma$, let $\mathcal{P}_2$ be a plane perpendicular to the vector $\boldsymbol \vgamma(\mathbf{x}_2)$. Let  $\mathcal{S}_1 \subset \mathcal{P}_1$ and $\mathcal{S}_2 \subset \mathcal{P}_2$ be compact sets such that
\begin{align}
    \max_{\mathbf{x}, \mathbf{y} \in  \mathcal{S}_1} \lvert\lvert \mathbf{x} - \mathbf{y} \rvert \rvert = \max_{\mathbf{x}, \mathbf{y} \in  \mathcal{S}_2} \lvert\lvert \mathbf{x} - \mathbf{y} \rvert \rvert = \delta >0.
\end{align}
We then can define a tube $\Phi(\delta)$ with two side boundaries $\mathcal{S}_1$ and $\mathcal{S}_2$ and $\Gamma \subset \Phi(\delta)$. 

Recall that $\boldsymbol \gamma_\epsilon ( \mathbf{x}) = \pi_\epsilon(\mathbf{x})^{-1} \mathbf{J}[ \pi_\epsilon(\mathbf{x}) ] $. By the stationary Fokker-Planck equation, we have that
\begin{align}
    \nabla \cdot ( \boldsymbol \gamma_\epsilon ( \mathbf{x}) \pi_\epsilon(\mathbf{x})) = 0, \quad \text{for all} \ \mathbf{x} \in \mathbb{R}^n.
\end{align}
Furthermore, by the Gauss's theorem,
\begin{align} \label{gauss}
    \int_{\mathcal{S}} \big( \boldsymbol\gamma_\epsilon ( \mathbf{x}) \pi_\epsilon(\mathbf{x}) \cdot \mathbf{n} \big) \rd \mathcal{S}   =   \int_{\mathcal{V}}   \big( \nabla \cdot (\boldsymbol\gamma_\epsilon ( \mathbf{x}) \pi_\epsilon(\mathbf{x}) \big) \rd \mathcal{V} = 0,
\end{align}
where $\mathcal{S}$ in the surface of $\Phi(\delta)$, $\mathbf n$ is the outward normal vector to  $\mathcal{S}$, and $\mathcal{V}$ is the volume of $\Phi(\delta)$. 

For the left hand side of Eq. \eqref{gauss}, it can be written as a sum of three terms
\begin{align} 
    \int_{\mathcal{S}} \big( \boldsymbol \gamma_\epsilon ( \mathbf{x}) \pi_\epsilon(\mathbf{x}) \cdot \mathbf{n} \big) \rd \mathcal{S} 
    &=   \int_{\mathcal{S}_1} \big( \boldsymbol\gamma_\epsilon ( \mathbf{x}) \pi_\epsilon(\mathbf{x}) \cdot \mathbf{n} \big) \rd \mathcal{S}_1   + \int_{\mathcal{S}_2} \big( \boldsymbol\gamma_\epsilon ( \mathbf{x}) \pi_\epsilon(\mathbf{x}) \cdot \mathbf{n} \big) \rd \mathcal{S}_2  \notag \\
    &+ \int_{\mathcal{S}_3} \big( \boldsymbol \gamma_\epsilon ( \mathbf{x}) \pi_\epsilon(\mathbf{x}) \cdot \mathbf{n} \big) \rd \mathcal{S}_3 , \label{eq:surface2}
\end{align}
in which $\mathcal{S}_3$ is the lateral surface of $\Phi(\delta)$ and $\mathcal{S}_3 \cap \Gamma = \emptyset$. 

For every $\mathbf{y} \in \Gamma $, we denote $\mathcal{S}_\vy :=\Phi(\delta) \cap \mathcal{P}_\vy$, $\mathcal{P}_\vy$ is the plane perpendicular to $\vgamma(\vy)$. By Lemma \ref{lemma-2},  
\begin{align} \label{wkb-approx-cycle}
\frac{ \int_{\mathcal{S}_\mathbf{y}} f(\mathbf{x}) e^{ \frac{-\varphi (\mathbf{x})  }{\epsilon}}  \rd \mathcal{S}_\mathbf{y}}{\int_{\mathcal{S}_\mathbf{y}}  e^{ \frac{-\varphi (\mathbf{x})  }{\epsilon}}   \rd \mathcal{S}_\mathbf{y}}
= f(\mathbf{y}) + O(\epsilon),
\end{align}
for any continuous and bounded function $f: \mathbb{R}^n \rightarrow \mathbb{R}$. Furthermore, by the definition of function $v$ in Lemma \ref{thm:speed-omega}, we can approximate the ratio of two integrals 
\begin{align} \label{normal.factor.approx2}
    \frac{  \int_{\mathcal{S}_\mathbf{y}}  e^{ \frac{-\varphi (\mathbf{x})  }{\epsilon}}   \rd \mathcal{S}_\mathbf{y}   }{\int_{\mathcal{S}_1}  e^{ \frac{-\varphi (\mathbf{x})  }{\epsilon}}   \rd \mathcal{S}_1 } = \frac{\sqrt{2\pi \epsilon v(\mathbf{y})  } e^{ \frac{-\varphi (\vy)  }{\epsilon}} }{  \sqrt{2\pi \epsilon v(\vx)  } e^{ \frac{-\varphi (\vx_1)  }{\epsilon}} } + O(\epsilon)
    = \frac{ \sqrt{v(\mathbf{y})}  }{ \sqrt{v(\mathbf{x}_1)} } + O(\epsilon),
\end{align}
in which we use Laplace's method in the first equality and $\varphi \equiv 0$ on $\Gamma$   in the second equality. To choose $f(\mathbf{x}) = \omega({\mathbf{x}})$ for Eq.  \eqref{wkb-approx-cycle}, combined with the result of \eqref{normal.factor.approx2}, we can obtain
\begin{align}  \label{normal.factor.approx3}
    \frac{  \int_{\mathcal{S}_\mathbf{y}}  \omega(\mathbf{x}) e^{ \frac{-\varphi (\mathbf{x})  }{\epsilon}}   \rd \mathcal{S}_\mathbf{y}   }{\int_{\mathcal{S}_1}  e^{ \frac{-\varphi (\mathbf{x})  }{\epsilon}}   \rd \mathcal{S}_1 } = \frac{ \omega(\mathbf{y})\sqrt{v(\mathbf{y})}  }{ \sqrt{v(\mathbf{x}_1)} } + O(\epsilon).
\end{align}
Since the function $e^{ -\frac{\varphi (\mathbf{x})  }{\epsilon}}$ is concentrated near $\Gamma$, we can further approximate the normalization factor $\int_{\mathbb{R}^n} \omega(\mathbf{x}) e^{ \frac{-\varphi (\mathbf{x})  }{\epsilon}}   \rd \mathbf{x}$ as follows
\begin{align} \label{normal.factor.approx}
\frac{\int_{\mathbb{R}^n} \omega(\mathbf{x}) e^{ \frac{-\varphi (\mathbf{x})  }{\epsilon}}   \rd \mathbf{x}}{\int_{\mathcal{S}_1}  e^{ \frac{-\varphi (\mathbf{x})  }{\epsilon}}   \rd \mathcal{S}_1 }  = 
\frac{\int_{\Gamma}  \int_{\mathcal{S}_\mathbf{y}} \omega(\mathbf{x}) e^{ \frac{-\varphi (\mathbf{x})  }{\epsilon}}   \rd \mathcal{S}_\mathbf{y}  \rd\mathbf{\mathbf{y}} }{\int_{\mathcal{S}_1}  e^{ \frac{-\varphi (\mathbf{x})  }{\epsilon}}   \rd \mathcal{S}_1 } + O(\epsilon).
\end{align}
By  \eqref{normal.factor.approx3} and \eqref{normal.factor.approx}, we have that
\begin{align} \label{normal.factor.approx4}
    \frac{\int_{\mathbb{R}^n} \omega(\mathbf{x}) e^{ \frac{-\varphi (\mathbf{x})  }{\epsilon}}   \rd \mathbf{x}}{\int_{\mathcal{S}_1}  e^{ \frac{-\varphi (\mathbf{x})  }{\epsilon}}   \rd \mathcal{S}_1 } = \frac{\int_{\Gamma}\omega(\mathbf{y}) \sqrt{v(\mathbf{y})}  \rd\mathbf{\mathbf{y}} }{\sqrt{v(\mathbf{x}_1)}}
 + O(\epsilon).
\end{align}
By the WKB expansion of $\pi_\epsilon$ in Eq. \eqref{wkb:stationary}, with Eq. \eqref{normal.factor.approx4}, the first term on the right hand side of Eq. \eqref{eq:surface2} can be written as
\begin{align} \label{wkb-approx-cycle2}
    \int_{\mathcal{S}_1} \big( \boldsymbol\gamma_\epsilon ( \mathbf{x}) \pi_\epsilon(\mathbf{x}) \cdot \mathbf{n} \big)  \rd \mathcal{S}_1 &= \frac{\int_{\mathcal{S}_1} \left( \boldsymbol\gamma_\epsilon( \mathbf{x})  \cdot \mathbf{n} \right)  \omega(\mathbf{x}) e^{ \frac{- \varphi (\mathbf{x})  }{\epsilon}}   \rd \mathcal{S}_1 }{  \int_{\mathbb{R}^n} \omega(\mathbf{x}) e^{ \frac{-\varphi (\mathbf{x})  }{\epsilon}}   \rd \mathbf{x} } \notag \\
    &= \left(\frac{\sqrt{v(\mathbf{x}_1)}}{ \int_{\Gamma} \omega(\mathbf{y}) \sqrt{v(\mathbf{y})}  \rd\mathbf{y}}  \right) \left( \frac{  \int_{\mathcal{S}_1} \left( \boldsymbol\gamma_\epsilon( \mathbf{x})  \cdot \mathbf{n} \right)   \omega(\mathbf{x}) e^{ \frac{-\varphi(\mathbf{x})  }{\epsilon}}    \rd \mathcal{S}_1  }{  \int_{\mathcal{S}_1}  e^{ \frac{-\varphi(\mathbf{x})  }{\epsilon}}   \rd \mathcal{S}_1   } \right) + O(\epsilon) .
\end{align}
Note that $\boldsymbol\gamma_\epsilon(\mathbf{x}) \rightarrow \boldsymbol\gamma(\mathbf{x})$. Without loss of generality, we assume that $\boldsymbol\gamma(\mathbf{x}_1)$ is inflow and $\boldsymbol\gamma(\mathbf{x}_2)$ is outflow of $\Phi(\delta)$. To choose $f(\mathbf{x}) = \left( \boldsymbol\gamma_\epsilon( \mathbf{x})  \cdot \mathbf{n} \right)   \omega(\mathbf{x})$ for Eq.  \eqref{wkb-approx-cycle}, combined with Eq. \eqref{wkb-approx-cycle2}, we then obtain  
\begin{align} \label{first.term}
     \int_{\mathcal{S}_1} \big( \boldsymbol\gamma_\epsilon ( \mathbf{x}) \pi_\epsilon(\mathbf{x}) \cdot \mathbf{n} \big)   \rd \mathcal{S}_1 \rightarrow - C   \sqrt{v(\mathbf{x}_1)} \omega(\mathbf{x}_1) \lvert\lvert  \boldsymbol\gamma(\mathbf{x}_1) \rvert \rvert \quad \text{as} \quad \epsilon \rightarrow 0, 
\end{align}
in which the constant $C = 1/ \int_{\Gamma} \omega(\mathbf{y}) \sqrt{v(\mathbf{y})}   \rd\mathbf{y} $. By the same approach, the second term on the right hand side of Eq. \eqref{eq:surface2} has a convergence
\begin{align} \label{second.term}
     \int_{\mathcal{S}_2} \big( \boldsymbol\gamma_\epsilon ( \mathbf{x}) \pi_\epsilon(\mathbf{x}) \cdot \mathbf{n} \big)  \rd \mathcal{S}_2 \rightarrow  C   \sqrt{v(\mathbf{x}_2)} \omega(\mathbf{x}_2) \lvert\lvert  \boldsymbol\gamma(\mathbf{x}_2) \rvert \rvert \quad \text{as} \quad \epsilon \rightarrow 0. 
\end{align}
Since $\mathcal{S}_3 \cap \Gamma = \emptyset$, the third term
\begin{align} \label{third.term}
     \int_{\mathcal{S}_3} \big( \boldsymbol\gamma_\epsilon ( \mathbf{x}) \pi_\epsilon(\mathbf{x}) \cdot \mathbf{n} \big)   \rd \mathcal{S}_3 \rightarrow 0 \quad \text{as} \quad \epsilon \rightarrow 0.
\end{align}
To apply the results \eqref{first.term}, \eqref{second.term}, and \eqref{third.term} to the equations \eqref{gauss} and \eqref{eq:surface2}, we can show that
\begin{align} \label{constant.circle}
    \Big| C  \sqrt{v(\mathbf{x}_1)} \omega(\mathbf{x}_1) \lvert\lvert  \boldsymbol\gamma(\mathbf{x}_1) \rvert\rvert -  C  \sqrt{v(\mathbf{x}_2)} \omega(\mathbf{x}_2) \lvert\lvert  \boldsymbol\gamma(\mathbf{x}_2)  \rvert\rvert    \Big| = 0.
\end{align}
Since Eq. \eqref{constant.circle} holds for every pair of two points on $\Gamma$, $ \sqrt{v(\mathbf{x})} \omega(\mathbf{x}) \lvert\lvert  \boldsymbol\gamma(\mathbf{x}) \rvert \rvert$ is constant on $\Gamma$. 

For $\vx \in \Gamma$, the marginal density can be approximated by
\begin{align} \label{eq:margial.approx2}
    g_\epsilon(\mathbf{x}) = \frac{ \int_{\mathbb{R}^n  \backslash \Gamma}  \omega(\mathbf{y}) e^{ \frac{-\varphi (\mathbf{y})  }{\epsilon}}  \rd \mathbf{y} }{ \int_{\mathbb{R}^n} \omega(\mathbf{y}) e^{ \frac{-\varphi (\mathbf{y})  }{\epsilon}}   \rd \mathbf{y} } = \frac{\int_{\mathcal{S}_\vx} \omega(\mathbf{y}) e^{ \frac{- \varphi (\mathbf{y})  }{\epsilon}} \rd \mathcal{S}_\vx }{ \int_{\mathbb{R}^n} \omega(\mathbf{y}) e^{ \frac{-\varphi (\mathbf{y})  }{\epsilon}}   \rd \mathbf{y} } + O(\epsilon) =   \frac{ \omega(\mathbf{x}) \sqrt{v(\mathbf{x})}}{ \int_{\Gamma}  \omega(\mathbf{y}) \sqrt{v(\mathbf{y})} \rd\mathbf{y}} + O(\epsilon), 
\end{align}
which follows the steps from  Eq.  \eqref{wkb-approx-cycle} to Eq. \eqref{wkb-approx-cycle2} with choosing $\boldsymbol\gamma_\epsilon( \mathbf{x})  \cdot \mathbf{n} = 1$. Furthermore, since $\sqrt{v(\vx)} \omega(\mathbf{x}) \lvert\lvert  \gamma(\mathbf{x}) \rvert \rvert$ is constant on $\Gamma$, with the result \eqref{eq:margial.approx2}, there exists a constant $K$ such that
\begin{align}
  g_\epsilon(\mathbf{x})\lvert\lvert \boldsymbol\gamma(\mathbf{x}) \rvert\rvert = K + O(\epsilon).
\end{align}

\end{proof}

\begin{acknowledgements}
We thank Lowell Thompson and Ying-Jen Yang for many helpful discussions. 
\end{acknowledgements}

% Authors must disclose all relationships or interests that 
% could have direct or potential influence or impart bias on 
% the work: 
%
\section*{Conflict of interest}
The authors declare that they have no conflict of interest.

% BibTeX users please use one of
%\bibliographystyle{spbasic}      % basic style, author-year citations
%\bibliographystyle{spmpsci}      % mathematics and physical sciences
%\bibliographystyle{spphys}       % APS-like style for physics
%\bibliography{}   % name your BibTeX data base
\bibliographystyle{unsrt}
\bibliography{SLC.bib}

\end{document}